\documentclass[11pt]{article}

\usepackage{latexsym}
\usepackage{epsfig}
\usepackage{cite}
\usepackage{amssymb}
\usepackage{amsthm}
\usepackage{amsmath}
\usepackage{graphicx}
\usepackage{paralist}
\usepackage{verbatim,color}
\usepackage[colorlinks,urlcolor=blue,citecolor=blue,linkcolor=blue]{hyperref}

\newcommand{\ignore}[1]{}
\renewcommand{\abstract}{\section*{Abstract}}

\newcommand{\E}{\mathrm{E}}

\newtheorem{theorem}{Theorem}[section]

\newtheorem{lemma}[theorem]{Lemma}

\newtheorem{claim}[theorem]{Claim}

\newtheorem{definition}[theorem]{Definition}

\newtheorem{corollary}[theorem]{Corollary}

\newcommand{\bd}{{\bf d}}

\newcommand{\bI}{{\bf I}}

\newcommand{\cC}{{\cal C}}
\newcommand{\cD}{{\cal D}}
\newcommand{\cE}{{\cal E}}
\newcommand{\cF}{{\cal F}}

\newcommand{\cU}{{\cal U}}

\newcommand{\cW}{{\cal W}}
\newcommand{\cZ}{{\cal Z}}

\newcommand{\EX}{\hbox{\bf E}}

\newcommand{\Sec}[1]{\hyperref[sec:#1]{\S\ref*{sec:#1}}} 
\newcommand{\Eqn}[1]{\hyperref[eq:#1]{(\ref*{eq:#1})}} 
\newcommand{\Fig}[1]{\hyperref[fig:#1]{Fig.\,\ref*{fig:#1}}} 
\newcommand{\Tab}[1]{\hyperref[tab:#1]{Tab.\,\ref*{tab:#1}}} 
\newcommand{\Thm}[1]{\hyperref[thm:#1]{Thm.\,\ref*{thm:#1}}} 
\newcommand{\Lem}[1]{\hyperref[lem:#1]{Lem.\,\ref*{lem:#1}}} 
\newcommand{\Prop}[1]{\hyperref[prop:#1]{Prop.~\ref*{prop:#1}}} 
\newcommand{\Cor}[1]{\hyperref[cor:#1]{Cor.~\ref*{cor:#1}}} 
\newcommand{\Def}[1]{\hyperref[def:#1]{Defn.~\ref*{def:#1}}} 
\newcommand{\Alg}[1]{\hyperref[alg:#1]{Alg.~\ref*{alg:#1}}} 
\newcommand{\Ex}[1]{\hyperref[ex:#1]{Ex.~\ref*{ex:#1}}} 
\newcommand{\Clm}[1]{\hyperref[clm:#1]{Claim~\ref*{clm:#1}}} 

\newcommand{\algo}{{{\sc MinBucket}}}
\newcommand{\cm}{ECM}

\title{Why do simple algorithms for triangle enumeration work in the real world?}

\author{Jonathan W. Berry \thanks{Sandia National Laboratories, Albuquerque. \{\texttt{jberry},\texttt{caphill}\}\texttt{@sandia.gov}.
This manuscript has been authored by Sandia Corporation under Contract No. DE-AC04-94AL85000 with the U.S. Department of Energy.  The United States Government retains and the publisher, by accepting the article for publication, acknowledges that the United States Government retains a non-exclusive, paid-up, irrevocable, world-wide license to publish or reproduce the published form of this manuscript, or allow others to do so, for United States Government purposes.} \and Luke A. Fostvedt \thanks{Iowa State University, Ames, Iowa 50011. \{\texttt{fostvedt},\texttt{dnordman},\texttt{agw}\}\texttt{@iastate.edu}} \and
Daniel J. Nordman \footnotemark[2]  \and Cynthia A. Phillips \footnotemark[1] \and C. Seshadhri \thanks{Sandia National Laboratories, Livermore. \texttt{scomand@sandia.gov}} \and Alyson G. Wilson \thanks{North Carolina State University. \{\texttt{alyson\_wilson@ncsu.edu}\}}
}

\addtocounter{page}{-1}
\date{}
\begin{document}
\maketitle
\thispagestyle{empty}

\abstract{Listing all triangles is a fundamental graph operation.  Triangles can have important interpretations
in real-world graphs, especially social and other interaction networks.
Despite the lack of provably efficient (linear, or slightly super-linear) worst-case algorithms for this problem, practitioners
run simple, efficient heuristics to find all triangles in graphs with millions of vertices.
How are these heuristics exploiting
the structure of these special graphs to provide major speedups in running time?

We study one of the most prevalent algorithms used by practitioners. A trivial algorithm
enumerates all paths of length $2$, and checks if each such path is incident to a triangle.
A good heuristic is to enumerate only those paths of length $2$
where the middle vertex has the lowest degree.
It is easily implemented and is empirically known to give remarkable speedups over the trivial algorithm.

We study the behavior of this algorithm over graphs with heavy-tailed
degree distributions, a defining feature of real-world graphs. The
erased configuration model (ECM) efficiently generates a graph
with asymptotically (almost) any desired degree sequence.
We show that the expected running time of this algorithm
over the distribution of graphs created by the ECM is controlled by the $\ell_{4/3}$-norm of the degree
sequence.
Norms of the degree sequence are a measure of the heaviness of the tail, and it is
precisely this feature that allows non-trivial speedups of simple
triangle enumeration algorithms.  As a corollary of our main theorem,
we prove expected linear-time performance for degree sequences
following a power law with exponent $\alpha \geq 7/3$, and non-trivial
speedup whenever $\alpha \in (2,3)$.  }

\newpage

\section{Introduction} Finding triangles in graphs is a classic theoretical
problem with numerous practical applications. The recent explosion of
work on social networks has led to a great interest in fast algorithms
to find triangles in graphs. The social sciences and physics communities often
study triangles in real networks and use them to reason about underlying social
processes~\cite{Co88,Po98,WaSt98, Burt04, Burt07, FoDeCo10}.
Much of the information about triangles in the last four papers is determined by a complete enumeration of all triangles in a (small) graph.
Triangle enumeration is also a fundamental subroutine for other more complex algorithmic tasks \cite{BerryHLP11, FudosH97}.

From a theoretical perspective, Itai and Rodeh \cite{ItRo78} gave algorithms for triangle finding
in $O(n^\omega)$ time (where $n$ is the number of vertices and $\omega$ is the matrix
multiplication constant) using fast matrix multiplication.
Vassilevska Williams and Williams \cite{VaWi10} show deep connections between matrix multiplication and
(edge-weighted) triangle enumeration.  But much of this
work is focused on dense graphs. Practitioners usually deal with massive sparse graphs with large variance in degrees,
where sub-quadratic time algorithms can be trivially obtained, but are still too slow to run.

Practioners enumerate triangles on massive graphs (with millions of vertices) using fairly
simple heuristics, which are often easily parallelizable.
This work is motivated by the following question: \emph{can we theoretically explain why simple algorithms for triangle enumeration work in the real world?}

Consider a trivial algorithm. Take an undirected graph with $n$ vertices, $m$ edges, and
degree sequence $d_1, d_2, \ldots, d_n$ (so the degree of vertex $v$ is
$d_v$). Call a path of length $2$ ($P_2$)
\emph{closed} if it participates in a triangle and \emph{open} otherwise. Simply enumerate all
$P_2$s and output the closed ones. The total running time is $\Theta(\sum_v d^2_v)$
(assume that checking if a $P_2$ is closed can be done in constant time), since every $P_2$ involves
a pair of neighbors for the middle vertex.
We will henceforth refer to this as \emph{the} {trivial algorithm}.
A simple heuristic is to only enumerate paths where the middle vertex has the lowest
degree among of the $3$ vertices in the path.
We denote this algorithm by \algo.

\begin{asparaenum}
\item Create $n$ empty buckets $B_1, B_2, \ldots, B_n$.
\item For each edge $(u,v)$: if $d_u \leq d_v$, place it in $B_u$, otherwise place it in $B_v$.  Break ties consistently.
\item For each bucket $B_v$: iterate over all $P_2$s formed by edges in $B_v$ and output closed ones.
\end{asparaenum}
\algo{} is quite common in practice (sometimes taking the somewhat
strange name \emph{nodeIterator++})
and has clean parallel implementations with no
load balancing issues \cite{ScWa05,Co09,SuVa11}. For such simple algorithms, the total work pretty much determines the parallel runtime.
For example, it would take $n$ processors with perfect
speed up running a $\Theta(n^2)$-work algorithm to compete with
a single processor running a $\Theta(n)$-work algorithm.

\algo{} is often the algorithm of choice for
triangle enumeration because of its simplicity
and because it beats the trivial algorithm by orders of magnitude,
as shown in the previous citations. (A quick check shows at least 60 citations to \cite{Co09}, mostly
involving papers that deal with massive scale graph algorithms.)
The algorithm itself has been discovered and rediscovered in various forms over the past decades. The earliest reference the authors could find was from the mid-80s where Chiba and Nishizeki
\cite{ChNi85} devised a sequential version of the above algorithm. We provide a more detailed history later.

Nonetheless, \algo{} has a poor worst-case behavior.
It would perform terribly on a high degree regular bipartite graph. If the input sparse graph (with high variance
in degree) simply consisted of many such bipartite graphs of varying
sizes, \algo{} would perform no better than its trivial cousin. Then why is it good in practice?

\subsection{Results} \label{sec:results}

Since the seminal results of Barab\'{a}si and Albert~\cite{BarabasiAlbert99}, Faloutsos et al~\cite{FFF99}, Broder et al~\cite{BrKu+00},
researchers have assumed that massive graphs obtained from the real world have \emph{heavy-tailed degree distributions} (often approximated as a power law). The average degree
is thought to be a constant (or very slowly growing), but the variance is quite large.
The usual approximation is to think of the number of vertices of degree $d$
as decaying roughly as $1/d^\alpha$ for some small constant $\alpha$.

This seems to have connections with \algo. If edges tend to connect vertices of fairly disparate degrees
(quite likely in a graph with large variance in degrees), \algo{} might provably
give good running times. This is exactly what we set out to prove, for a natural
distribution on heavy-tailed graphs.

Consider any list of positive integers $\bd = (d_1, d_2, \ldots, d_n)$, which
we think of as a ``desired" degree sequence. In other words, we wish to construct
a graph on $n$ vertices where vertex $v \in [n]$ has degree $d_v$.
The \emph{configuration model} (CM) \cite{bc1978, b1980, MoRe98, n2003} creates a random graph that almost
achieves this. Imagine vertex $v$ being incident to $d_v$ ``stubs", which can be thought
of as half-edges. We take a random perfect matching between the stubs, so pairs of stubs
are matched to each other. Each such pair creates an edge, and we end up with a multigraph
with the desired degree sequence. Usually, this is converted to a simple graph
by removing parallel edges and self-loops\cite{bdm2006}. We refer to this graph
distribution as $\cm(\bd)$, for input degree sequence $\bd$.
This model has a fairly long history (which
we relegate to a later section) and is a standard method to construct a graph with a desired degree sequence. It is closely connected to models given
by Chung and Lu \cite{ChLu02, ChLuVu03} and Mihail and Papadimitriou~\cite{MiPa02},
in the context of eigenvalues of graphs with a given degree sequence. These models
simply connect vertices $u$ and $v$ independently with probability $d_u d_v/2m$, similarly to the Erd\H{o}s-R\'{e}nyi construction.

Our main theorem gives a bound on the expected running time of \algo{} for $\cm(\bd)$.  We set $m = (\sum_v d_v)/2$.
We will henceforth assume that $0 < d_1 \leq d_2 \ldots \leq d_n$ and that $d_n < \sqrt{m}/2$.
This ``truncation" is a standard assumption for analysis of the configuration model \cite{MoRe98, ChLu02, MiPa02, ChLuVu03, n2003, bdm2006}.
We use $\sum_v$ as a shorthand for $\sum_{i=1}^n$, since
it is a sum over all vertices. The run time bottleneck for \algo{} is in $P_2$ enumeration,
and checking whether a $P_2$ is closed is often assumed to be a constant time operation.
Henceforth, when we say ``running time," we mean the number of $P_2$s enumerated.

\begin{theorem} \label{thm:main} Consider a degree sequence $\bd = (d_1, d_2, \ldots, d_n)$, where
$m = (\sum_v d_v)/2$ and $d_n < \sqrt{m}/2$. The expected (over $\cm(\bd)$) number of $P_2$s enumerated by \algo{} is
$O(n + m^{-2}(\sum_v d_v^{4/3})^3)$.
\end{theorem}

(Our main theorem applies to Chung-Lu graphs as well. Details are given in Appendix~\ref{app:cl}.)
Before we actually make sense of this theorem, let us look at a corollary of this theorem.
It has been repeatedly observed that
degree sequences in real graphs have heavy tails, often approximated as a \emph{power law} \cite{BarabasiAlbert99}.
Power laws say something about the moments of the degree distribution (equivalently, norms of the degree sequence).
Since it does not affect our main theorem or corollary, we
choose a fairly loose definition of power law. This is a binned version of the usual definition, which states the number
of vertices of degree $d$ is proportional to $n/d^\alpha$. (Even up to constants, this is never precisely true because
there are many gaps in real degree sequences.)

\begin{definition} \label{def:power} A degree sequence $\bd$ satisfies
a power law of exponent $\alpha > 1$ if the following holds
for all $k \leq \log_2 d_n - 1$: for $d = 2^k$, the number of sequence terms in $[d,2d]$ is
$\Theta(n/d^{\alpha-1})$.
\end{definition}

The following shows an application of our theorem for common values of $\alpha$.  This bound is tight as we show in Section~\ref{sec:tight}.
(When $\alpha \geq 3$, the trivial algorithm runs in linear time because $\sum_v d^2_v = O(n)$.)

\begin{corollary} \label{cor:power} Suppose a degree sequence $\bd$ (with largest term $< \sqrt{m}/2$) satisfies a power law with exponent $\alpha \in (2,3)$.
Then the expected running time of \algo{} of $\cm(\bd)$ is asymptotically better than the trivial algorithm,
and is \emph{linear} when $\alpha > 7/3$.
\end{corollary}

\subsection{Making sense of \Thm{main}} \label{sec:sense}

First, as a sanity check, let us actually show that \Thm{main} beats the trivial bound, $\sum_v d^2_v$. This
is a direct application of H\"{o}lder's inequality for conjugates $p=3$ and $q=3/2$.
$$ (\sum_{v} d_v^{4/3})^3 = (\sum_{v} d_v^{2/3}\cdot d_v^{2/3})^3
\leq \Big(\sum_v d_v^{\frac{2}{3} \cdot 3}\Big)^{3 \cdot \frac{1}{3}} \Big(\sum_v d_v^{\frac{2}{3} \cdot \frac{3}{2}}\Big)^{3 \cdot \frac{2}{3}}
= (2m)^2(\sum_v d^2_v) $$
Rearranging, we get $m^{-2} (\sum_{v} d_v^{4/3})^3 = O(\sum_v d^2_v)$, showing that our bound
at least holds the promise of being non-trivial.

Consider the uniform distribution on the vertices.
Assuming $m > n$, we can write our running time bound as $n(\EX[d^{4/3}_v])^3$, as opposed to the trivial
bound of $\sum_v d^2_v = n\EX[d^2_v]$. If the degree ``distribution" (think of the random
variable given by the degree of a uniform random vertex) has a small $4/3$-moment, the running time is small. This
can happen even though the second moment is large, and this is where
\algo beats the trivial algorithm. In other words, if the tail of the degree sequence
is heavy but not too heavy, \algo{} will perform well.

And this is exactly what happens when $\alpha > 2$ for power law degree sequences.
When $\alpha > 7/3$, the $4/3$-moment becomes constant and the running time is linear.
(It is known that for ECM graphs over power law degree sequences with $\alpha > 7/3$, the clustering
coefficient (ratio of triangles to $P_2$s) converges to zero~\cite{n2003}.)
We show in \Sec{tight} that the running time bound achieved in the following corollary for power laws with $\alpha > 2$ is tight.
When $\alpha \leq 2$, \algo{} gets no asymptotic improvement over the trivial algorithm.
For convenience, we will drop the big-Oh notation, and replace it by $\lessdot$. So $A \lessdot B$ means $A = O(B)$.

\begin{proof} (of \Cor{power}) First, let us understand the trivial bound.
Remember than $d_n$ is the maximum degree.
$$ \sum_v d^2_v \lessdot \sum_{k=1}^{\log_2n - 1} (n/2^{k(\alpha-1)}) 2^{2k}
= n \sum_{k=1}^{\log_2n - 1} 2^{k(3-\alpha)} \lessdot n + n d^{3-\alpha}_n $$
We can argue that the expected number of wedges enumerated by the trivial algorithm is $\Omega(\sum_v d^2_v)$ (\Clm{trivial}).
Now for the bound of \Thm{main}.
$$ m^{-2}(\sum_v d^{4/3}_v)^3 \lessdot n^{-2}\Big(\sum_{k=1}^{\log_2n - 1} (n/2^{k(\alpha-1)}) 2^{4k/3}\Big)^3
= n \Big(\sum_{k=1}^{\log_2n - 1} 2^{k(7/3-\alpha)}\Big)^3 \lessdot n + n d^{7-3\alpha}_n $$
Regardless of $d_n$, if $\alpha > 7/3$, then the running time of \algo{} is linear. Whenever
$\alpha \in (2,3)$, $7-3\alpha < 3-\alpha$, and \algo{} is asymptotically faster than a trivial enumeration.
\end{proof}

\subsection{Significance of \Thm{main}} \label{sec:sig}

\Thm{main} connects the running time of a commonly used algorithm to the norms of the degree sequences,
a well-studied property of real graphs.
So this important property of heavy-tails in real graphs allows for the algorithmic benefit of \algo{}.
We have discovered that for a fairly standard graph
model inspired by real degree distributions, \algo{} is very efficient.

We think of this theorem as a proof of concept: theoretically showing that a common property
of real world inputs allows for the efficient performance of a simple heuristic.
Because of our distributional assumptions as well as bounds on $\alpha$, we agree with
the (skeptical) reader that this does not
fully explain why \algo{} works in the real world\footnote{As the astute reader would have noticed,
our title is a question, not a statement.}. Nonetheless, we feel that this makes progress towards
that, especially for a question that is quite hard to formalize. After all, there is hardly
any consensus in the social networks community on what real graphs look like.

But the notion that distinctive properties of real world graphs can be used to prove
efficiency of simple algorithms is a useful way of thinking.
This is
one direction to follow for going beyond worst-case analysis.
Our aim here is not to
design better algorithms for triangle enumeration, but to give a theoretical argument for why
current algorithms do well.

The proof is obtained (as expected) through various probabilistic arguments bounding
the sizes of the different buckets. The erased configuration model, while easy to implement
and clean to define, creates some niggling problems for analysis of algorithms. The edges
are not truly independent of each other, and we have to take care of these weak dependencies.

Why the $4/3$-norm? Indeed, that is one of the most surprising features of this result (especially
since the bound is tight for power laws of $\alpha > 2$). As we bound the buckets sizes and make
sense of the various expressions, the total running time is expressed as a sum of various degree terms.
Using appropriate approximations, it tends to 
rearrange into norms of the degree sequence. Our proof goes over two sections.
We give various probabilistic calculations for the degree behavior in \Sec{prob}, which
set the stage for the run-time accounting.
In \Sec{cohen}, we start bounding bucket sizes and finally get to the $4/3$-moment.
In \Sec{tight}, we show that bounds achieved in the proof of \Cor{power} are tight. This is mostly
a matter of using the tools of the previous sections.  In \Sec{careful}, we give a tighter analysis that
gives an explicit expression for strong upper bounds on running time and in \Sec{experiments} we experimentally show these
more careful bounds closely approximate the expected runtime of ECM graphs, with runtime constants under $1$ for graphs up to 80M nodes.

\section{Related Work} \label{sec:related}

The idea of using some sort of degree binning, orienting edges, or thresholding for finding and enumerating triangles
has been used in many results. Chiba and Nishizeki~\cite{ChNi85} give bounds for a sequential
version of \algo{} using the degeneracy of a graph. This does not give bounds for \algo{}, although their algorithm
is similar in spirit. Alon, Yuster, and Zwick~\cite{AlYuZw97}
find triangles in $O(m^{1.41})$ using degree thresholding and matrix multiplication
ideas from Itai and Rodeh~\cite{ItRo78}.
Chrobak and Eppstien \cite{ChrobakE91} use
acyclic orientations for linear time triangle enumeration in planar graphs.
Vassilevska Williams and Williams \cite{VaWi10} show that fast algorithms
for weighted triangle enumeration leads to remarkable consequences, like faster all-pairs shortest paths.
In the work most closely to ours, Latapy \cite{latapy08} discusses various triangle
finding algorithms, and also focuses on power-law graphs.
He shows the trivial bound of $O(mn^{1/\alpha})$ when the power law exponent is $\alpha$.
Essentially, the maximum degree is $n^{1/\alpha}$ and that directly gives a bound
on the number of $P_2$s.

\algo{} has received attention from various experimental studies.
Schank and Wagner \cite{sw2005} perform an
experimental study of many algorithms, including a sequential version of \algo{}
which they show to be quite efficient.
Cohen~\cite{Co09} specifically describes \algo{} in the context of Map-Reduce.
Suri and Vassilvitskii~\cite{SuVa11} do many experiments on real graphs in Map-Reduce
and show major speedups (a few orders of magnitude) for \algo{} over the trivial
enumeration. Tsourakakis~\cite{t2008} gives a good survey of various methods used in practice
for triangle counting and estimation.

Explicit triangle enumerations have been used for various applications
on large graphs. Fudos and Hoffman~\cite{FudosH97} use triangle enumeration for a graph-based approach
for solving systems of geometric constraints.
Berry et al~\cite{BerryHLP11} touch every triangle as part of their community detection algorithm for
large graphs.

Configuration models for generating random graphs with given degree sequences
have a long history. Bender and Canfield~\cite{bc1978} study this model for counting
graphs with a given degree sequence. Wormald~\cite{Wo81} looks at the connectivity
of these graphs. Molloy and Reed~\cite{MoRe95, MoRe98} study various
properties like the largest connected component of this graph distribution.
Physicists studying complex networks have also paid attention to this model \cite{NeStWa01}.
Britton, Deijfen, and Martin-L\"{o}f \cite{bdm2006} show that the simple
graph generated by the ECM
asymptotically matches the desired degree sequence.
Aiello, Chung, and Lu \cite{AiChLu01} give a model for power-law graphs, where edge $(u,v)$
are independent inserted with probability $d_u d_v/2m$. This was studied
for more general degree sequences in subsequent work by Chung, Lu, and Vu \cite{ChLu02,ChLuVu03}.
Mihail and Papadimitriou \cite{MiPa02} independently discuss this model. Most
of this work focused on eigenvalues and average distances in these graphs.
Newman~\cite{n2003} gives an excellent survey of these models, their similarities, and applications.

\section{Degree behavior of $\cm(\bd)$} \label{sec:prob}

We fix a degree sequence $\bd$ and focus on the distribution $\cm(\bd)$. All expectations
and probabilities are over this distribution.
Because of dependencies in the erased configuration model,
we will need to formalize our arguments carefully.
We first state a general lemma giving a one-sided tail bound for dependent random variables
with special conditional properties. The proof is in the appendix.

\begin{lemma} \label{lem:mart} Let $Y_1, Y_2, \ldots, Y_k$ be independent random variables,
and $X_i = f_i(Y_1, Y_2, \ldots, Y_i)$ be $0$-$1$ random variables. Let $\alpha \in [0,1]$.
Suppose $\Pr[X_1] \geq \alpha$ and $\Pr[X_i=1 | Y_1, Y_2, \ldots, Y_{i-1}] \geq \alpha$ for all $i$.
Then, $\Pr[\sum_{i=1}^k X_i < \alpha k \delta] < \exp(-\alpha k (1-\delta)^2/2)$ for any $\delta\in(0,1)$.
\end{lemma}

We now prove a tail bound on degrees of vertices; the probability the degree of vertex $v$ deviates by a constant factor of $d_v$
is $\exp(-\Omega(d_v))$. Let $\beta, \beta', \delta,\delta'$ denote sufficiently small constants.

Before we proceed with our tail bounds, we
describe a process to construct the random matching of stubs. We are interested in
a particular vertex $v$.
Order the stubs such that the $d_v$ $v$-stubs are in the beginning; the remaining
stubs are ordered arbitrarily. We start with the first stub, and match to a uniform random
stub (other than itself). We then take the next unmatched stub, according
to the order, and match to a uniform random unmatched stub. And so on and so forth.
The final connections are clearly dependent, though the choice among
unmatched stubs is done independently. This is formalized as follows.
Let $Y_i$ be an independent uniform random integer in $[1,2m-2(i-1)-1]$. This
represents the choice at the $i$th step, since in the $i$th step, we have
exactly $2m - 2(i-1)-1$ choices. Imagine that we first draw these independent $Y_i$'s.
Then we deterministically construct the matching on the basis of these numbers. (So the first
stub is connected to the $Y_1$st stub, the second unmatched stub is connected to the
$Y_2$nd unmatched stub, etc.)

\begin{lemma} \label{lem:lower-tail} Assume $d_n < \sqrt{m}/2$. Let $D_v$ be the random variable
denoting the degree of $v$ in the resulting graph.  There exist sufficiently small constants $\beta, \beta' \in (0,1)$,
such that $ \Pr[D_v < \beta' d_v] < \exp(-\beta d_v)$.
\end{lemma}

\begin{proof} Suppose $d_v>1$. We again order the stubs so that the $d_v$ $v$-stubs are in the beginning.
Let $X_j$ be the indicator random variable for the $j$th matching
forming a new edge with $v$. Note that $\sum_{j=1}^{\lfloor d_v/2 \rfloor} X_j \leq D_v$.
Observe that $X_j$ is a function of $Y_1, Y_2, \ldots, Y_j$.
Consider any $Y_1, Y_2, \ldots, Y_{j-1}$ and suppose the matchings created
by these variables link to vertices $v_0=v,v_1, v_2, \ldots, v_{j-1}$ (distinct)
such that there are $n_j$ links to vertex $v_j$ such that $\sum_{i=0}^{j-1} n_i=(j-1)$.
Then, for $j=1,\ldots,\lfloor d_v/2 \rfloor$,
\begin{eqnarray*}
\EX[X_j | Y_1, Y_2, \ldots, Y_{j-1}]  &\geq  &1 - \frac{(d_v-j-n_0) + \sum_{1\leq i \leq j-1; n_i \neq 0} (d_{v_i} - n_i)}{2m-2(j-1)-1}\\
&\geq  &1 - \frac{ -2(j-1)-1 + \sum_{i=0}^{j-1}d_{v_i} }{2m-2(j-1)-1} \\
&\geq  &1 - \frac{\sum_{i=0}^{j-1}d_{v_i}}{2m-2d_v}.
\end{eqnarray*}
Note that $\sum_{i=0}^{j-1}d_{v_i} \leq (\sqrt{m}/2)^2 = m/4$, by the bound on the maximum degree.
We also get $2m - 2d_v > m$, so we bound $\EX[X_j | Y_1, Y_2, \ldots, Y_{j-1}] \geq 3/4$.
By \Lem{mart} (setting $\delta = 2/3$ and bounding $\alpha k > d_v/4$), \[
\Pr[D_v < d_v/8] \leq \Pr\left[\sum_{j=1}^{\lfloor d_v/2 \rfloor} X_j < d_v/8\right]
\leq \Pr\left[\sum_{j=1}^{\lfloor d_v/2 \rfloor} X_j < \lfloor d_v/2 \rfloor/2 \right]
< \exp(-d_v(1/3)^2/8).\]
\end{proof}

This suffices to prove the trivial bound for the trivial algorithm.

\begin{claim} \label{clm:trivial} The expected number of wedges enumerated by the trivial algorithm
is $\Omega(\sum_v d^2_v)$.
\end{claim}

\begin{proof} The expected number of wedges enumerated is $\Omega(\sum_v D^2_v)$, where $D_v$
is the actual degree of $v$. Using \Lem{lower-tail}, $\EX[D^2_v] = \Omega(d^2_v)$.
\end{proof}

We will need the following basic claim about the joint probability of two edges.

\begin{claim} \label{clm:pair} Let $v, w, w'$ be three distinct vertices.
The probability that edges $(v,w)$ and $(v,w')$ are present in the final graph is at most $d^2_v d_w d_{w'}/m^2$.
\end{claim}

\begin{proof} Assume $d_v>1$. Let
$C_{v,w}$ be the indicator random variable for edge $(v,w)$ being
present (likewise define $C_{v,w'}$).  Label the stubs of each vertex as
$s_1^v,\ldots,s_{d_v}^v$;   $s_1^w,\ldots,s_{d_w}^w$; and $s_1^{w'},\ldots,s_{d_{w'}}^{w'}$.
Let $C_{s_i^v,s_j^w}$ be the indicator random variable for edge being
present between stubs $s_i^v$ and $s_j^w$ (likewise define $C_{s_i^v,s_j^{w'}}$).  Then the event
$\{C_{v,w}C_{v,w'}=1\}$ that edges $(v,w)$ and $(v,w')$ are present    is a subset of the event $
\cup_{1 \leq i\neq j \leq d_v}\cup_{k=1}^{d_w} \cup_{\ell=1}^{d_{w'}} \{C_{s_i^v,s_k^w}C_{s_j^v,s_\ell^{w'}}=1\}$.
Hence,
\[\Pr[C_{v,w}C_{v,w'} = 1] \leq \sum_{1 \leq i\neq j \leq d_v}\sum_{k=1}^{d_w} \sum_{\ell=1}^{d_{w'}} \Pr[C_{s_i^v,s_k^w}C_{s_j^v,s_\ell^{w'}}=1].\]

Fix $1 \leq i \neq j \leq d_v$, $1\leq k \leq d_w$ and $1 \leq \ell \leq d_{w'}$ and order stubs $s_i^v,s_j^v$ first in
the ECM wiring.  Then, $\Pr[C_{s_i^v,s_k^w}C_{s_j^v,s_\ell^{w'}}=1] =  \Pr[C_{s_i^v,s_k^w} = 1] \Pr[C_{s_j^v,s_\ell^{w'}} = 1 | C_{s_i^v,s_k^w} = 1]$ where  $\Pr[C_{s_i^v,s_k^w} = 1] = [2m-1]^{-1}$ and  $\Pr[C_{s_j^v,s_\ell^{w'}} = 1 | C_{s_i^v,s_k^w} = 1]=[2m-3]^{-1}$.  Hence,
\[
\Pr[C_{v,w}C_{v,w'} = 1] \leq d_v(d_v-1)d_w d_{w'}/m^2
\]
using $(2m-1)(2m-3)\geq m^2$ when $m \geq 3$.
\end{proof}

\section{Getting the $4/3$ moment} \label{sec:cohen}

We will use a series of claims to express the running time of \algo{}
in a convenient form. For vertex $v$, let $X_v$ be the random variable denoting
the number of edges in $v$'s bin. The expected running time is at most $\EX[\sum_v X_v(X_v-1)]$. This is because number of wedges in each
bin is ${X_v\choose 2} \leq X^2_v -X_v$.

We further break $X_v$ into the sum $\sum_w Y_{v,w}$, where $Y_{v,w}$ is
the indicator for edge $(v,w)$ being in $v$'s bin.
As mentioned earlier, $C_{v,w}$ is the indicator for edge $(v,w)$ being present.
Note that $Y_{v,w} \leq C_{v,w}$, since $(v,w)$ can only
be in $v$'s bin if it actually appears as an edge.

We list out some bounds on expectations. Only the second one really
uses the binning of \algo{}.

\begin{claim} \label{clm:arms} Consider vertices $v,w,w'$ ($w \neq w'$).
\begin{itemize}
\item $\EX[Y_{v,w}Y_{v,w'}] \leq d^2_vd_wd_{w'}/m^2$.
\item There exist sufficient small constants $\delta, \delta' \in (0,1)$ such that: if $d_w < \delta d_v$ then $\EX[Y_{v,w}Y_{v,w'}] \leq 2\exp(-\delta' d_v) d^2_vd_wd_{w'}/m^2$.
\end{itemize}
\end{claim}

\begin{proof}  We use the trivial bound of $Y_{v,w}Y_{v,w'} \leq C_{v,w}C_{v,w'}$.
By \Clm{pair}, $\EX[Y_{v,w}Y_{v,w'}] \leq \EX[C_{v,w}C_{v,w'}] \leq d^2_vd_wd_{w'}/m^2$.

Now for the interesting bound.
The quantity $\EX[Y_{v,w}Y_{v,w'}]$
is the probability that both $Y_{v,w}$ and $Y_{v,w'}$ are $1$. For this to
happen, we definitely require both $(v,w)$ and $(v,w')$ to be present as edges.
Call this event $\cE$. We also require (at the very least) the degree of $v$
to be at most the degree of $w$ (otherwise the edge $(v,w)$ will not
be put in $v$'s bin.) Call this event $\cF$.  If $D_v, D_w$ denote
the degrees of $v$ and $w$, note that $D_w \leq d_w < \delta d_v$, implying  event $\cF$ is contained in
the event $\{D_v < \delta d_v\}$ when $d_w < \delta d_v$.  Hence, the event $Y_{v,w}Y_{v,w'} = 1$
is contained in $\cE \cap \{D_v < \delta d_v\}$.   Assume $d_v>2,d_w>0,d_{w'}>0$ or else $\EX[Y_{v,w}Y_{v,w'}]=0$ trivially when $\delta<1/2$.

As in the proof of \Clm{pair}, let $C_{s_i^v,s_j^w}$ be the indicator random variable for edge being
present between stubs $s_i^v$ and $s_j^w$ of vertices $v,w$ (and analogously  define $C_{s_i^v,s_j^{w'}}$).
Then $\cE$ is contained in $
\cup_{1 \leq i\neq j \leq d_v}\cup_{k=1}^{d_w} \cup_{\ell=1}^{d_{w'}} \{C_{s_i^v,s_k^w}C_{s_j^v,s_\ell^{w'}}=1\}$
so that
\begin{eqnarray*}
\Pr[Y_{v,w}Y_{v,w'} = 1] &\leq &\Pr[\cE,  D_v  < \delta d_v ] \\
&\leq &\sum_{1 \leq i\neq j \leq d_v}\sum_{k=1}^{d_w} \sum_{\ell=1}^{d_{w'}} \Pr[C_{s_i^v,s_k^w}C_{s_j^v,s_\ell^{w'}}=1, D_v  < \delta d_v]\\
&=& \sum_{1 \leq i\neq j \leq d_v}\sum_{k=1}^{d_w} \sum_{\ell=1}^{d_{w'}} \Pr[C_{s_i^v,s_k^w}C_{s_j^v,s_\ell^{w'}}=1] \Pr[D_v  < \delta d_v | C_{s_i^v,s_k^w}C_{s_j^v,s_\ell^{w'}}=1].
\end{eqnarray*}
Given fixed values of $i,j,k,\ell$ and  order stubs $s_i^v,s_j^v$ first in
the ECM wiring. Then, $\Pr[C_{s_i^v,s_k^w}C_{s_j^v,s_\ell^{w'}}=1] \leq m^{-2}$ as  in the proof of \Clm{pair}. Additionally, conditioned on $C_{s_i^v,s_k^w}C_{s_j^v,s_\ell^{w'}}=1$, the remaining
stubs form an ECM with respect to a new degree sequence formed by
replacing $2m,d_v,d_w,d_{w'}$ in the original degree sequence
by  $2\tilde{m}=2m-4,d_v-2,d_w-1,d_{w'}-1$.  Let $\tilde{D}_v $ denote the degree of $v$
in the final graph from the new degree sequence. Then, conditioned on $C_{s_i^v,s_k^w}C_{s_j^v,s_\ell^{w'}}=1$,
$D_v=2+\tilde{D}_v$ so that conditional probability is bounded by
\begin{eqnarray*}
\Pr[D_v  < \delta d_v|C_{s_i^v,s_k^w}C_{s_j^v,s_\ell^{w'}}=1] &=&  \Pr[\tilde{D}_v  < \delta d_v-2|C_{s_i^v,s_k^w}C_{s_j^v,s_\ell^{w'}}=1] \\
&\leq& \Pr[\tilde{D}_v  < \delta (d_v-2)|C_{s_i^v,s_k^w}C_{s_j^v,s_\ell^{w'}}=1]\\
&\leq &2\exp(-\delta' d_v)
\end{eqnarray*}
since $\delta<1$.  That is, \Lem{lower-tail} applies to $\tilde{D}_v$ with respect to the new degree sequence where
$v$ has degree $d_v-2$ and    each degree in this new sequence is less than $\sqrt{\tilde{m}}/2$ by assumption.  The bound  $\Pr[Y_{v,w}Y_{v,w'} = 1] \leq 2\exp(-\delta' d_v) d_v^2d_w d_{w'}/m^2$ then follows.
\end{proof}

Armed with these facts, we can bound the expected number of $P_2$s contained in a single bucket.

\begin{lemma} \label{lem:xv}
$$\EX[X_v(X_v-1)] = O\Big(\exp(-\delta d_v) d^2_v + \Big(m^{-2}{d^2_v} \sum_{\substack{w: \\ d_w \geq \delta d_v}} \sum_{\substack{w \neq w':\\ d_{w'} \geq \delta d_v}} d_w d_{w'}\Big)\Big) $$.
\end{lemma}

\begin{proof} We will write out
\begin{eqnarray*}
X^2_v  = \ (\sum_w Y_{v,w})^2  & = & \sum_w Y^2_{v,w} + \sum_w\sum_{w' \neq w} Y_{v,w} Y_{v,w'}
\end{eqnarray*}
where $\sum_w Y^2_{v,w}=\sum_w Y_{v,w}=X_v$ as each $Y_{v,w}$ is a 0-1 variable. Hence,
\begin{eqnarray*}
\EX[X_v(X_v-1)] = \sum_w\sum_{w' \neq w} \EX[Y_{v,w} Y_{v,w'}] & \leq &
\sum_{\substack{w: \\ d_w \geq \delta d_v}} \sum_{\substack{w \neq w':\\ d_{w'} \geq \delta d_v}} \EX[Y_{v,w} Y_{v,w'}]
+ \sum_{\substack{w:\\ d_w < \delta d_v}} \sum_{w' \neq w} \EX[Y_{v,w} Y_{v,w'}] \\
& & + \sum_{\substack{w':\\ d_{w'} < \delta d_v}} \sum_{w \neq w'} \EX[Y_{v,w} Y_{v,w'}] \\
& = &  \sum_{\substack{w: \\ d_w \geq \delta d_v}} \sum_{\substack{w \neq w':\\ d_{w'} \geq \delta d_v}} \EX[Y_{v,w} Y_{v,w'}]
+ 2\sum_{\substack{w:\\ d_w < \delta d_v}} \sum_{w' \neq w} \EX[Y_{v,w} Y_{v,w'}]\\
&\leq & \frac{d^2_v}{m^2} (\sum_{w : d_w \geq \delta d_v} d_w)^2 + 2\sum_{\substack{w:\\ d_w < \delta d_v}} \sum_{w' \neq w} \EX[Y_{v,w} Y_{v,w'}]
\end{eqnarray*}
by splitting the  sums into cases, $d_w \geq \delta d_v$ and $d_w < \delta d_v$, and  using the trivial bound of \Clm{arms}
for the first quantity.
We satisfy the conditions to use the second part of \Clm{arms} as
w%
\begin{eqnarray*}
\sum_{\substack{w:\\ d_w < \delta d_v}} \sum_{w' \neq w} \EX[Y_{v,w} Y_{v,w'}] & \leq & 2\sum_{\substack{w:\\ d_w < \delta d_v}} \sum_{w' \neq w} \exp(-\delta d_v) d^2_vd_wd_{w'}/m^2 \\
& \leq & \exp(-\delta d_v) d^2_v,
\end{eqnarray*}
where $\sum_{i=1}^n d_i=2m$,
\end{proof}

With this bound for $\EX[X_v(X_v-1)]$, we are ready to prove
\Thm{main}.

\begin{theorem} \label{thm:sum-xv}
$\EX[\sum_v X_v(X_v-1)] = O(n + m^{-2}(\sum_{i=1}^n d_i^{4/3})^3)$.
\end{theorem}

\begin{proof} We use linearity of expectation and sum the bound in \Lem{xv}. Note
that $\exp(-\delta d_v) d^2_v$ is a decreasing function of $d_v$ and is hence $O(1)$.
The double summation of \Lem{xv} can be upper bounded by $(\sum_{w: d_w \geq \delta d_v} d_w)^2$.
$$ \EX[\sum_v X_v(X_v-1)] \lessdot n + m^{-2}\sum_v d^2_v\big(\sum_{w: d_w \geq \delta d_v} d_w\big)^2
= n + m^{-2}\sum_v \sum_{w: d_w \geq \delta d_v} \sum_{w': d_{w'} \geq \delta d_v} d_v^2 d_w  d_{w'}$$
This is the moment where the $4/3$ moment will appear.
Since $d_w \geq \delta d_v$ and $d_{w'} \geq \delta d_v$,
$d^{2/3}_v \leq \delta^{-2/3} d^{1/3}_w d^{1/3}_{w'}$.
Therefore, $d_v^2 d_w d_{w'} = d^{4/3}_v d^{2/3}_v d_w d_{w'}$
$\leq \delta^{-2/3} (d_v d_w d_{w'})^{4/3}$. Wrapping it up,
\begin{eqnarray*}
m^{-2}\sum_v \sum_{w: d_w \geq \delta d_v}
\sum_{w': d_{w'} \geq \delta d_v} d_v^2 d_w  d_{w'}&  \lessdot & m^{-2}\sum_v \sum_{w: d_w \geq \delta d_v}
\sum_{w': d_{w'} \geq \delta d_v}  (d_v d_w d_{w'})^{4/3}\\
& \lessdot &    m^{-2} (\sum_v d_v^{4/3})^3.
\end{eqnarray*}
\end{proof}

\section{Proving tightness} \label{sec:tight}

We show that the bound achieved by \Thm{main} is tight for power laws with $\alpha > 2$.
This shows that the bounds given in the proof of \Cor{power} are tight.
The proof, as expected, goes by reversing most of the inequalities given earlier.
For convenience, we will assume for the lower bound that $d_n < \sqrt{m}/4$, instead
of the $\sqrt{m}/2$ used for the upper bound. This makes for cleaner technical arguments
(we could just as well prove it for $\sqrt{m}/2$, at the cost of more pain). Proofs are
given in Appendix~\ref{app:tight}.

\begin{claim} \label{clm:tight} Let $\bd$ be a power law degree sequence with $\alpha \in (2, 7/3)$
with $d_n < \sqrt{m}/4$.
Then the expected number of $P_2$s enumerated by \algo{} over $\cm(\bd)$ is $\Omega(nd^{7-3\alpha}_n)$.
\end{claim}

\section{Tighter bounds on the running time}
\label{sec:careful}

Under a specific choice of degrees, we can pin down the running time of \algo{} up to lower order terms.
Rather than starting with an arbitrary degree sequence, we draw the degree for each vertex independently at random from a reference degree
distribution $\cD$, given by pdf $f$.  Specifically, $f(d)$ is the probability that a node draws degree value $d$, for $d$ and integer in $[0, \infty)$.
After nodes draw degree values, the rest of the ECM construction proceeds as described in \Sec{results}.

Formally, let $\cD_n$ be the distribution with support $\{1,2,\ldots,\lfloor\sqrt{n}/\log^2n\rfloor\}$,
where the probability of $d$ is proportional to $f(d)$. Note that we do not allow a degree of $0$
and cap the max degree at $\sqrt{n}/\log^2n$ (instead of $\sqrt{n}$). These are mostly
conveniences for a cleaner proof. We pick the degree sequence by taking $n$ i.i.d.
draws from $\cD_n$. So, the degree sequence $\bd$ is distributed according to the product $\cD^n_n$.
Then we generate an ECM with $\bd$.
For convenience, we denote $1 - 1/\sum_{d \leq n} f(d)$ by $\gamma_n$. Note that $\gamma_n \rightarrow 0$, as $n \rightarrow \infty$.
The probability of $d$ under $\cD_n$ is $f(d)(1-\gamma_n)$. We use $m = \sum_v d_v/2$ to denote
the number of edges in the \emph{multigraph} and heavily use $m \geq n/2$.

Our analysis assumes that when an edge joins two vertices of the same degree, the edge is placed in the bucket for both edges.  Thus we slightly
overcount the work for \algo.  Let  $X_{i,n}$ be the size of the bucket for an arbitrary node $i$ in a graph generated by ECM with $n$ nodes.
We wish to bound the expected triangle-searching work $\EX[\sum_{i=1}^n \binom{X_{i,n}}{2}]$ in an ECM graph, as the number of nodes $n\rightarrow \infty$.
We denote $r$th moment, $r>0$, of the reference degree distribution  $f$ as $\EX[d^r] = \sum_{t=1}^\infty t^r\cdot f(t)$.
The main theorem is as follows.

\begin{theorem} \label{thm:limit} Fix any $n$ and a degree distribution $\cD$ such that $\EX[d]$ and $\EX[d^{4/3}]$ are finite.
Then
\[
\lim_{n \rightarrow \infty} \frac{1}{n}\EX\left[\sum_{i=1}^n {X_{i,n} \choose 2}\right]   =   \frac{1}{2 (\EX[d])^2}   \sum_{t_1=1}^\infty \sum_{t_2=t_1}^\infty \sum_{t_3=t_1}^\infty t_1(t_1-1)t_2 t_3 f(t_1)f(t_2)f(t_3) \in (0,\infty).
\]
\end{theorem}

Note that the expectation for the running time is over two ``sources" of randomness, the degree sequence, and the actual configuration graph.
Throughout this section, we use $o(1)$ to denote a quantity that goes to zero as $n \rightarrow \infty$.
We use the shorthand $A = B \pm C$ for $A \in [B - C, B+C]$. This is often done with $C = o(1)$.

The overall structure of the proof can be thought of as a three-stage process.
First, we fix the degree sequence $d_1, d_2, \ldots, d_n$ (and number of vertices $n$)
and bound $\sum_{i=1}^n {X_{i,n} \choose 2}$. This is analogous to what was done earlier in \Sec{cohen}.
We require a more careful analysis where we keep track of various constant factors, so that
the final limit can be precisely stated. We then take expectations over the degree sequence.
Finally, we take the limit $n \rightarrow \infty$.

\subsection{Bin sizes for low degree vertices} \label{sec:fixed}

We fix the degree sequence $\bd = (d_1, d_2, \ldots, d_n)$ and consider the probability space of
$\cm(\bd)$. As before, we use $D_v$ to denote the degree of $v$ in the resulting
simple graph. We denote the $u$-stubs (for all vertices $u$)
by $u_1, u_2, \ldots u_{d_u}$. Since $n$ is fixed here, we drop the subscript $n$ from $X_{v,n}$.

The main lemma of this section is a tight bound on the wedge count in $v$'s bin, when $d_v$ is small.
This is a tigher analogue of \Lem{xv}.

\begin{lemma} \label{lem:xv-sharp} Suppose $d_v \leq \log n$. $$ \EX[X_v(X_v-1)] =
(1 \pm o(1)) d_v(d_v-1)\Big( \sum_{w \neq v: d_w \geq d_v} \sum_{w' \notin \{v,w\}: d_{w'} \geq d_v}  d_w d_{w'}\Big)/(2m)^2 \pm o(1) $$
\end{lemma}

We first prove numerous intermediate claims about the following events. We use $c$ to denote
a large enough constant.
\smallskip
\begin{asparaitem}
\item $\cE_v$: The event that $D_v \neq d_v$.
\item $\cF_{v,b}$: Let $b$ be a number. This is the event that $D_v < b$.
\item $\cU_{u_i,v_j}$: For stubs $u_i, v_j$, this is the event that the $u_i$ is matched to $v_j$ in the multigraph.
\item $\cW_{v,w,w'}$: The event that wedge $\{(v,w), (v,w')\}$ is formed.
\item $\cZ_{v,w,w'}$: The event that wedge $\{(v,w), (v,w')\}$ is in $v$'s bin.
\end{asparaitem}
Note that
$\cW_{v,w,w'} = \bigcup_{i \neq j \leq d_v, k \leq d_w, \ell \leq d_{w'}} \Big(\cU_{v_i,w_k} \cap \cU_{v_j, w'_\ell}\Big).$
It will be convenient to denote the index set $\{(i,j,k,\ell) | i \neq j \leq d_v, k \leq d_w, \ell \leq d_{w'}\}$
as $\bI$.

We begin by getting handles on $\cE_v, \cF_{v,b}$.

\begin{claim} \label{clm:cond-deg} Fix distinct vertices $v, w, w'$, distinct stubs $v_i, v_j, w_k, w'_\ell$, and an arbitrary vertex $u$.
\begin{asparaitem}
\item For any $u$ such that $d_u < c\log n$, $Pr[\cE_u | \cU_{v_i,w_k} \cap \cU_{v_j,w'_\ell}] < (\log^3 n)/\sqrt{m}$.
\item For any $u$ and $b$ such that $d_u > b > \log n$, then $Pr[\cF_{u,b} | \cU_{v_i,w_k} \cap \cU_{v_j,w'_\ell}] < (\log^3 n)/\sqrt{m}$.
\end{asparaitem}
\end{claim}

\begin{proof} Conditioning on the event $E \equiv \cU_{v_i,w_k} \cap \cU_{v_j,w'_\ell}$   means
that these stubs are matched, with $2m-4$ stubs remaining in the multigraph;  the event $E$ itself creates no self-loops or multi-edges for vertex $u$, even if $u\in\{v,w,w'\}$, as $v,w,w'$ are distinct.
Let $S \equiv \{u_1,\ldots,u_{d_u}\}\setminus \{v_i,w_k, v_j,w'_\ell\}$.  For any stubs $u_i,u_j \in S$ of vertex $u$, $i \neq j$, the event $A_{u_i,u_j} \equiv $``$u_i,u_j$ pair" has
probability $P(A_{u_i,u_j}|E) \leq (2m-5)^{-1}$, and the event $B_{u_i,u_j,z} \equiv$``$u_{i},u_j$ pair to stubs
of vertex $z\neq u$" has  probability  $P(B_{u_i,u_j,z}|E) \leq d_z (d_z-1)/[(2m-5)(2m-7)]$. Then, $Pr[\cE_u | \cU_{v_i,w_k} \cap \cU_{v_j,w'_\ell}]$ is at most
\[
Pr\left(  \left[\bigcup_{u_i \neq u_j \in S}A_{u_i,u_j}\right] \bigcup
\left[\bigcup_{u_i \neq u_j \in S, \atop z \neq u}B_{u_i,u_j,z}\right]
\Big| E\right) \leq \sum_{u_i \neq u_j \in S} Pr(A_{u_i,u_j}|E) + \sum_{u_i \neq u_j \in S, \atop z \neq u}Pr(B_{u_i,u_j,z}|E),
\]
which is bounded by  $(c \log n)^2 [2m-5]^{-1} +  (c \log n)^2 \sum_{z}  d_z^2/[(2m-5)(2m-7)]
<  (\log n)^3/\sqrt{m}$,
using $\sum_{z}  d_z^2 \leq \sqrt{n} 2 m $ and $ \sqrt{n/m} \leq \sqrt{2}$.

If $d_u < c\log n$, the second part is simply a consequence of the first part. If $d_u \geq c\log n$,
we can apply \Lem{mart} to bound the probability by $1/\sqrt{m}$.
\end{proof}

Our next step is to try to remove the conditioning on $\cU_{v_i,w_k} \cap \cU_{v_j,w'_\ell}$. We will
need the \emph{Boole-Bonferroni inequalities} (Prop. C.2 in~\cite{MoRa95}).

\begin{theorem} \label{thm:bb}  For any finite set of events $\cF_1, \cF_2, \ldots, \cF_r$,
$$ \sum_{i \leq r} \Pr[\cF_i] - \sum_{i < j \leq r} \Pr[\cF_i \cap \cF_j] \leq \Pr[\bigcup_{i \leq r} \cF_i] \leq \sum_{i \leq r} \Pr[\cF_i]$$
\end{theorem}

We have a general claim about probabilities of events in conjunction with $\cW_{v,w,w'}$.

\begin{claim} \label{clm:wed-upper} Fix an arbitrary event $\cC$. Suppose for all $(i,j,k,\ell) \in \bI$, $\Pr[\cC | \cU_{v_i,w_k} \cap \cU_{v_j, w'_\ell}] \in [B_L, B_U]$.
Then,
\begin{eqnarray*}
\Pr[\cC \cap \cW_{v,w,w'}] \leq B_U d_v(d_v-1) d_w d_{w'}/(2m-1)(2m-2) \\
\Pr[\cC \cap \cW_{v,w,w'}] \geq (B_L - 1/\log n)d_v(d_v-1) d_w d_{w'}/(2m-1)(2m-2)
\end{eqnarray*}
\end{claim}

\begin{proof} Applying \Thm{bb} to the events $\cC \cap \cU_{v_i,w_k} \cap \cU_{v_j, w'_\ell}$,
\begin{eqnarray} & & \sum_{(i,j,k,\ell) \in \bI} \Pr[\cC \cap \cU_{v_i,w_k} \cap \cU_{v_j, w'_\ell}]
- \sum_{(i,j,k,\ell) \neq (\hat{i},\hat{j},\hat{k},\hat{\ell}) \in \bI} \Pr[\cC \cap \cU_{v_i,w_k} \cap \cU_{v_j, w'_\ell} \cap \cU_{v_{\hat{i}},w_{\hat{k}}} \cap \cU_{v_{\hat{j}}, w'_{\hat{\ell}}}] \nonumber \\
& \leq & \Pr[\cC \cap \cW_{v,w,w'}] \leq \sum_{(i,j,k,\ell) \in \bI} \Pr[\cC \cap \cU_{v_i,w_k} \cap \cU_{v_j, w'_\ell}] \label{eq:bb}
\end{eqnarray}
We bound each sum separately.
\begin{eqnarray*}
\sum_{(i,j,k,\ell) \in \bI} \Pr[\cC \cap \cU_{v_i,w_k} \cap \cU_{v_j, w'_\ell}] & = & \sum_{(i,j,k,\ell) \in \bI} \Pr[\cC | \cU_{v_i,w_k} \cap \cU_{v_j, w'_\ell}] \Pr[\cU_{v_i,w_k} \cap \cU_{v_j, w'_\ell}] \\
& \leq & B_U \sum_{(i,j,k,\ell) \in \bI} \Pr[\cU_{v_i,w_k} \cap \cU_{v_j, w'_\ell}] \\
& = & B_U |\bI|/(2m-1)(2m-2) = B_U d_v(d_v-1) d_w d_{w'}/(2m-1)(2m-2)
\end{eqnarray*}
This completes the upper bound proof. By an identical argument, we get
$\sum_{(i,j,k,\ell) \in \bI} \Pr[\cC \cap \cU_{v_i,w_k} \cap \cU_{v_j, w'_\ell}] \geq B_L d_v(d_v-1) d_w d_{w'}/(2m-1)(2m-2)$.
We deal with the double summation in the next claim. Applying this bound to \Eqn{bb} completes the proof.
\end{proof}

\begin{claim} $\sum_{(i,j,k,\ell) \neq (\hat{i},\hat{j},\hat{k},\hat{\ell}) \in \bI} \Pr[\cC \cap \cU_{v_i,w_k} \cap \cU_{v_j, w'_\ell} \cap \cU_{v_{\hat{i}},w_{\hat{k}}} \cap \cU_{v_{\hat{j}}, w'_{\hat{\ell}}}] \leq (1/\log n) (d_v(d_v-1) d_w d_{w'}/m^2)$.
\end{claim}

\begin{proof} We can simply upper bound $\Pr[\cC \cap \cU_{v_i,w_k} \cap \cU_{v_j, w'_\ell} \cap \cU_{v_{\hat{i}},w_{\hat{k}}} \cap \cU_{v_{\hat{j}}, w'_{\hat{\ell}}}] \leq
\Pr[\cU_{v_i,w_k} \cap \cU_{v_j, w'_\ell} \cap \cU_{v_{\hat{i}},w_{\hat{k}}} \cap \cU_{v_{\hat{j}}, w'_{\hat{\ell}}}]$.

Consider $(i,j,k,\ell) \neq (\hat{i},\hat{j},\hat{k},\hat{\ell})$. The corresponding event involves stub pairs $(v_i,w_k)$, $(v_j, w'_\ell)$, $(v_{\hat{i}},w_{\hat{k}})$,
and $(v_{\hat{j}}, w'_{\hat{\ell}})$. If all these pairs are distinct (even if the stubs are common), then $\Pr[\cU_{v_i,w_k} \cap \cU_{v_j, w'_\ell} \cap \cU_{v_{\hat{i}},w_{\hat{k}}} \cap \cU_{v_{\hat{j}}, w'_{\hat{\ell}}}] \leq 1/m^4$. The number of pairs of tuples $(i,j,k,\ell) \neq (\hat{i},\hat{j},\hat{k},\hat{\ell})$ is at most $|\bI|^2$.

Since $(i,j,k,\ell) \neq (\hat{i},\hat{j},\hat{k},\hat{\ell})$, at most one of
the pairs can be the same. If (say) $(v_i,w_k) = (v_{\hat{i}},w_{\hat{k}})$, then $\Pr[\cU_{v_i,w_k} \cap \cU_{v_j, w'_\ell} \cap \cU_{v_{\hat{i}},w_{\hat{k}}} \cap \cU_{v_{\hat{j}}, w'_{\hat{\ell}}}] \leq 1/m^3$. The number of pairs of such $\bI$-tuples is at most $|\bI| (d_v d_w + d_v d_{w'})$.

\begin{eqnarray*} & & \sum_{(i,j,k,\ell), (\hat{i},\hat{j},\hat{k},\hat{\ell}) \in \bI} \Pr[\cU_{v_i,w_k} \cap \cU_{v_j, w'_\ell} \cap \cU_{v_{\hat{i}},w_{\hat{k}}} \cap \cU_{v_{\hat{j}}, w'_{\hat{\ell}}}]\\
& \leq & |\bI|^2/m^4 + |\bI| (d_v d_w + d_v d_{w'})/m^3
= (d_v(d_v-1) d_w d_{w'}/m^2) (d_v(d_v-1) d_w d_{w'}/m^2 + d_v(d_w + d_{w'})/m)
\end{eqnarray*}
Since $d_v, d_w, d_{w'} \leq \sqrt{n}/\log^2n$ and $m \geq n/2$, the final multiplier is at most $1/\log n$.
\end{proof}

As a corollary of \Clm{wed-upper}, we can set $\cC = \cW_{v,w,w'}$ and $B_L = B_U = 1$ to get the following.

\begin{claim} \label{clm:pr-wed} $\Pr[\cW_{v,w,w'}] = (1 \pm 1/\log n) d_v(d_v-1) d_w d_{w'}/(2m-1)(2m-2)$
\end{claim}
Now for an important lemma that bounds the probabilities of $\cZ_{v,w,w'}$.
\begin{lemma} \label{lem:cond-bound} Fix distinct vertices $v, w, w'$, such that $d_v < c\log n$.
\begin{asparaitem}
\item If $\min\{d_w, d_{w'}\} \geq d_v$, $\Pr[\cZ_{v,w,w'}] = (1 \pm o(1)) d_v(d_v-1) d_w d_{w'}/(2m)^2$.
\item If $\min\{d_w, d_{w'}\} < d_v$, $\Pr[\cZ_{v,w,w'}] = O((1/m^{1/4})\cdot d_v(d_v-1) d_w d_{w'}/m^2)$.
\end{asparaitem}
\end{lemma}

\begin{proof} Start with the first case of $\min\{d_w, d_{w'}\} \geq d_v$.
Since $\cZ_{v,w,w'} \subset \cW_{v,w,w'}$, $\Pr[\cZ_{v,w,w'}] \leq \Pr[\cW_{v,w,w'}]$ and \Clm{pr-wed}
completes the upper bound.

Consider the event $\cW_{v,w,w'} \cap \overline{\cE_v} \cap \overline{\cF_{w,d_v}} \cap \overline{\cF_{w',d_v}}$.
In this case, $D_v = d_v$ and $D_w, D_{w'} \geq d_v$, so $\cZ_{v,w,w'}$ contains this event.
\begin{eqnarray*}
\Pr[\cZ_{v,w,w'}] & \geq &
\Pr[\cW_{v,w,w'} \cap \overline{\cE_v} \cap \overline{\cF_{w,d_v}} \cap \overline{\cF_{w',d_v}}] \\
& = & \Pr[\cW_{v,w,w'} \cap \overline{\cE_v \cup \cF_{w,d_v} \cup \cF_{w',d_v}}] \\
& = & \Pr[\cW_{v,w,w'}] - \Pr[\cW_{v,w,w'} \cap (\cE_v \cup \cF_{w,d_v} \cup \cF_{w',d_v})] \\
& \geq & \Pr[\cW_{v,w,w'}] - \Pr[\cW_{v,w,w'} \cap \cE_v] - \Pr[\cW_{v,w,w'} \cap \cF_{w,d_v}] - \Pr[\cW_{v,w,w'} \cap \cF_{w',d_v}]
\end{eqnarray*}
(The last inequality follows by a union bound.)
We can bound $\Pr[\cW_{v,w,w'} \cap \cE_v]$ by applying the upper bound of \Clm{wed-upper}
with $B_U = (\log^3n)/\sqrt{m}$ (as obtained from \Clm{cond-deg}). This gives an upper bound of $((\log^3n)/\sqrt{m})d_v(d_v-1) d_w d_{w'}/(2m)^2$
$= o(1) \cdot d_v(d_v-1) d_w d_{w'}/(2m)^2$.
Identical arguments hold for $\Pr[\cW_{v,w,w'} \cap \cF_{w,d_v}]$ and $\Pr[\cW_{v,w,w'} \cap \cF_{w',d_v}]$.
Using the bound from \Clm{pr-wed} for $\Pr[\cW_{v,w,w'}]$, $\Pr[\cZ_{v,w,w'}] \geq (1-o(1))d_v(d_v-1) d_w d_{w'}/(2m)^2$.
This completes the first case.

Now, suppose $\min\{d_w, d_{w'}\} < d_v$. Note that $\cZ_{v,w,w'}$ is contained
in $\cW_{v,w,w'} \cap (\cE_v \cup \cE_w)$. This is because when $\overline{\cE_v \cup \cE_w}$
occurs, $D_v = d_v > d_w = D_w$, so the wedge cannot be present in $v$'s bin.
By the union bound, $\Pr[\cZ_{v,w,w'}] \leq \Pr[\cW_{v,w,w'} \cap \cE_v] + \Pr[\cW_{v,w,w'} \cap \cE_w]$.
Using the argument above, this is at most $((2\log^3n)/\sqrt{m})d_v(d_v-1) d_w d_{w'}/(2m)^2 = O((1/m^{1/4})\cdot d_v(d_v-1) d_w d_{w'}/m^2)$.
\end{proof}
We are ready to prove the main lemma.
\begin{proof} (of \Lem{xv-sharp}) Note that $\EX[X_v(X_v-1)] = \EX[\sum_{w \neq v} \sum_{w' \notin \{v,w\}} \mathbb{I}(\cZ_{v,w,w'})]$.\
By linearity of expectation, this sum is $\sum_{w \neq v} \sum_{w' \notin \{v,w\}} \Pr[\cZ_{v,w,w'}]$, which
can be split as follows.
\begin{eqnarray*}
\sum_{w \neq v} \sum_{w' \notin \{v,w\}} \Pr[\cZ_{v,w,w'}] & = &
\sum_{w \neq v: d_w \geq d_v} \sum_{w' \notin \{v,w\}: d_{w'} \geq d_v} \Pr[\cZ_{v,w,w'}]
+ \sum_{w \neq v: d_w \geq d_v} \sum_{w' \notin \{v,w\}: d_{w'} < d_v} \Pr[\cZ_{v,w,w'}]\\
& & + \sum_{w \neq v: d_w < d_v} \sum_{w' \notin \{v,w\}} \Pr[\cZ_{v,w,w'}]
\end{eqnarray*}
We deal with each of these summations using \Lem{cond-bound}. In the first summation, $\min\{d_w,d_{w'}\} \geq d_v$.
$$ \sum_{w \neq v: d_w \geq d_v} \sum_{w' \notin \{v,w\}: d_{w'} \geq d_v} \Pr[\cZ_{v,w,w'}]
= (1 \pm o(1)) d_v(d_v-1)\sum_{w \neq v: d_w \geq d_v} \sum_{w' \notin \{v,w\}: d_{w'} \geq d_v}  d_w d_{w'}/(2m)^2$$
In the second summation, $\min\{d_w,d_{w'}\} < d_v$.
$$ \sum_{w \neq v: d_w \geq d_v} \sum_{w' \notin \{v,w\}: d_{w'} < d_v} \Pr[\cZ_{v,w,w'}]
\lessdot d_v(d_v-1)/m^{2+1/4} \sum_w \sum_{w'} d_w d_{w'} \lessdot (\log n)^2/m^{1/4} = o(1)
$$
(We use the fact that $\sum_w d_w = 2m$, and the bound of $d_v < c\log n$.
The third summation can be handled similarly, completing the proof.
\end{proof}

\subsection{Expectation over $\cD_n$} \label{sec:exp}

Our aim is to take the expectation of \Lem{xv-sharp} over the degrees.
We will distinguish over the sources of randomness, by using $\EX_{\bd}[\ldots]$
to denote expectations over $\bd \sim \cD^n_n$. We use $\EX_G[\ldots]$ for the
expectation over the graph chosen from $\cm(\bd)$.
Because all vertices are basically identical, we just focus on the first vertex.
The main lemma is a fairly precise expression for the expectation of \Lem{xv-sharp}.
We denote the degree threshold $\sqrt{n}/\log^2n$ for $\cD_n$ by $M(n)$.

\begin{lemma} \label{lem:wedge-limit}
\begin{eqnarray*}
\EX_{\bd} \EX_G[X_1(X_1-1)] & = & (1 \pm o(1))/\EX[d_2]^2) \Big( \sum_{t_1 \leq \log n} \sum_{t_2 = t_1}^{M(n)}
\sum_{t_3 = t_1}^{M(n)} t_1 (t_1-1)t_2 t_3 f(t_1) f(t_2) f(t_3)\Big) \pm o(1) \\
& & + O\Big(\sum_{t_1 > \log n}^{M(n)} \sum_{t_2 = \delta t_1}^{M(n)} \sum_{t_3 = \delta t_1}^{M(n)} t^2_1 t_2 t_3 f(t_1) f(t_2) f(t_3)\Big)
\end{eqnarray*}
\end{lemma}

As a first step, we first condition on the choice of $d_1$
and choose the other degrees according to $\cD_n$. We denote the conditional expectation
over this distribution by $\EX_{\bd_{-1}}$. We will need the following
claim about the concentration of $m$.

\begin{claim} \label{clm:m} With probability $>1 - n^{-\log n}$, $|m - \EX_{\bd_{-1}}[m]| \leq n/\log n$.
\end{claim}

\begin{proof} We have $m = \sum_v d_v/2$ and $\EX_{\bd_{-1}} = d_1/2 + \sum_{v \neq 1} \EX_\bd[d_v/2]$.
Note that each $d_v/2$ is in $[1,\sqrt{n}/\log^2 n]$ and they are all independent.
By Hoeffding's inequality~\cite{Ho63}, $\Pr[|\sum_{v \neq 1} d_v/2 - \EX_\bd[\sum_{v \neq 1} d_v/2]| \geq n/\log n]
< 2 \exp(-2 (n/\log n)^2/\sum_{v \neq 1} (\sqrt{n}/\log^2 n)^2) = 2\exp(-2\log^2n) < n^{-\log n}$.
\end{proof}

As a step towards \Lem{wedge-limit}, we condition on $d_1$. When $d_1$ is small,
we can use \Lem{xv-sharp} of the previous section.

\begin{claim} \label{clm:cond-limit} Suppose $d_1 \leq \log n$. Then
$$ \EX_{\bd_{-1}} \EX_G[X_1(X_1-1)] = (1 \pm o(1)) d_1(d_1-1)/\EX[d_2]^2) \Big( \sum_{t_2 = t_1}^{M(n)} \sum_{t_3 = t_1}^{M(n)} t_2 t_3 f(t_2) f(t_3)\Big) \pm o(1) $$,
where  the order bound $o(1)$ applies uniformly over   $d_1 \leq \log n$.
\end{claim}

\begin{proof} It is convenient to define random variable $T_w$ where $T_w = 0$, if $d_w < d_1$,
and $T_w = d_w$ if $d_w \geq d_1$. We can rewrite the bound of \Lem{xv-sharp} as
$$ \EX_G[X_1(X_1-1)] =
(1 \pm o(1)) d_1(d_1-1)\Big( \sum_{w > 1} \sum_{w' \notin \{1,w\}}  T_w T_{w'}\Big)/(2m)^2 \pm o(1) $$
It will be convenient to denote the double summation by $A$.
Now, we take expectations over the $\bd_{-1} = (d_2, \ldots,  d_n)$. There is a slight technical difficulty,
since $m$ is itself a random variable. We use some Bayes' rule manipulations to handle this.
Let $\cC$ denote the event that $m \in [(1-1/\log n)\EX_{\bd_{-1}}[m], (1+1/\log n)\EX_{\bd_{-1}}[m]]$.
$$\EX_{\bd_{-1}} [d_1(d_1-1)A/m^2] = \EX_{\bd_{-1}} [d_1(d_1-1)A/m^2 | \cC] \Pr[\cC] + \EX_{\bd_{-1}} [d_1(d_1-1)A/m^2 | \overline{\cC}] \Pr[\overline{\cC}]$$
Since $d_1(d_1-1)A/m^2 \leq n^{4}$ and by \Clm{m}, $\Pr[\overline{\cC}] \leq n^{-\log n}$, the latter term is $o(1)$.
By definition of $\cC$, $\EX_{\bd_{-1}}[1/m^2 | \cC] = (1 \pm o(1))/\EX_{\bd_{-1}}[m]^2$.
\begin{eqnarray*}
\EX_{\bd_{-1}} \EX_G[X_1(X_1-1)] = (1 \pm o(1))d_1(d_1-1)\EX_{\bd_{-1}} [A | \cC]/(2\EX_{\bd_{-1}}[m])^2 \pm o(1)
\end{eqnarray*}
We defer the bound of $\EX_{\bd_{-1}} [A | \cC]$ to \Clm{A}. Let us first apply \Clm{A} to prove the main lemma.
\begin{eqnarray*}
& & \EX_{\bd_{-1}} \EX_G[X_1(X_1-1)] \\
& = & (1 \pm o(1))d_1(d_1-1)(n/2\EX_{\bd_{-1}}[m])^2 \sum_{t_2 = d_1}^{M(n)} \sum_{t_3 = d_1}^{M(n)} t_2 t_3 f(t_2) f(t_3) \pm o(1)
\pm d_1(d_1 - 1)/(2\EX_{\bd_{-1}}[m])^2
\end{eqnarray*}
Since $d_1 < \sqrt{n}/\log n$ and $\EX_{\bd_{-1}}[m] \geq n/2$, the final term is $o(1)$. Note that $2\EX_{\bd_{-1}}[m] = d_1 + \sum_{v > 1} \EX_{\bd_{-1}}[d_v]
= d_1 + (n-1)\EX[d_2] = (1 \pm o(1))n\EX[d_2]$. Plugging this bound in, the proof is completed.
\end{proof}

\begin{claim} \label{clm:A} $\EX_{\bd_{-1}}[A | \cC] = (1 \pm o(1)) n^2 \sum_{t_2 = d_1}^{M(n)} \sum_{t_3 = d_1}^{M(n)} t_2 t_3 f(t_2) f(t_3) \pm o(1)$,
where  the order bound $o(1)$ applies uniformly over   $d_1 \leq \log n$.
\end{claim}

\begin{proof} By Bayes' rule,	$\EX_{\bd_{-1}}[A | \cC] = (\Pr[\cC]^{-1})(\EX_{\bd_{-1}}[A] - \EX_{\bd_{-1}}[A | \overline{\cC}] \Pr[\overline{\cC}])$.
Since $A \leq n^4$ and $\Pr[\overline{\cC}] < n^{-\log n}$, $\EX_{\bd_{-1}}[A | \overline{\cC}] \Pr[\overline{\cC}] = o(1)$.
Therefore, $\EX_{\bd_{-1}}[A | \cC] = (1 \pm o(1))\EX_{\bd_{-1}}[A] - o(1)$. Writing out $A$,
\begin{eqnarray*}
\EX_{\bd_{-1}}\Big[\sum_{w > 1} \sum_{w' \notin \{1,w\}}  T_w T_{w'}\Big] = \sum_{w > 1} \sum_{w' \notin \{1,w\}} 	 \EX_{\bd_{-1}}[T_w] 	\EX_{\bd_{-1}}[T_{w'}]
\end{eqnarray*}
Because degrees are drawn independently, $\EX_{\bd_{-1}}[T_w T_{w'}] = \EX_{\bd_{-1}}[T_w] \EX_{\bd_{-1}}[T_{w'}]$.
$$ \EX_{\bd_{-1}}[T_w] = \EX_{\bd_{-1}}[T_{w'}] = (1-\gamma_n)\sum_{t_2 = d_1}^{M(n)} t_2 f(t_2) $$
Plugging this bound into the previous equation, we get
\begin{eqnarray*}
\EX_{\bd_{-1}}[A] = (n-1)(n-2)(1-\gamma_n)^2 \sum_{t_2 = d_1}^{M(n)} \sum_{t_3 = d_1}^{M(n)} t_2 t_3 f(t_2) f(t_3)
\end{eqnarray*}
We use the fact that $(n-1)(n-2) = (1 \pm o(1))n^2$ and $\gamma_n = o(1)$ to get the final proof.
\end{proof}

We require a bound for large $d_1$. This can be directly obtained with the looser arguments of \Lem{xv}.

\begin{claim} \label{clm:cond-limit2} Suppose $d_1 > \log n$.
$$ \EX_{\bd_{-1}} \EX_G[X_1(X_1-1)] = O\Big(d^2_1\sum_{t_2 = \delta d_1}^{M(n)}
\sum_{t_3 = \delta d_1}^{M(n)} t_2 t_3 f(t_2) f(t_3)\Big) + o(1),$$
where  the order bound $o(1)$ applies uniformly over   $d_1 > \log n$.
\end{claim}

\begin{proof} The first term in \Lem{xv} is $\exp(-\delta d_1)d^2_1$, which is $o(1)$ for $d_1 > \log n$.
(The constant $\delta$ comes from \Lem{xv}.)
Redefine $T_w$ to be $d_w$ is $d_w \geq \delta d_1$ and $0$ otherwise. Much of the following calculations
are similar to those in the proof above.
\begin{eqnarray*} \EX_{\bd_{-1}} \EX_G[X_1(X_1-1)] & \lessdot &
\EX_{\bd_{-1}}[m^{-2}d_1(d_1-1)\sum_{w > 1} \sum_{w' \notin \{1,w\}}  T_w T_{w'}] \pm o(1) \\
& \lessdot & n^{-2}d_1(d_1-1) \sum_{w > 1} \sum_{w' \notin \{1,w\}} \EX_{\bd_{-1}}[T_w] \EX_{\bd_{-1}}[T_{w'}] \\
& \lessdot & d^2_1\sum_{t_2 = \delta d_1}^{M(n)} \sum_{t_3 = \delta d_1}^{M(n)} t_2 t_3 f(t_2) f(t_3)
\end{eqnarray*}
\end{proof}

\Lem{wedge-limit} follows directly by applications of \Clm{cond-limit} and \Clm{cond-limit2}.
We simply express $\EX_{\bd} \EX_G[X_1(X_1-1)]$ as $\sum^{M(n)}_{t_1 = 1} f(t_1) \EX_{\bd_{-1}} \EX_G [X_1(X_1- 1) | d_1 = t_1]$.
When $t_1 \leq \log n$, we apply \Clm{cond-limit}. Otherwise, we use \Clm{cond-limit2}.

\subsection{Taking the limit} \label{sec:limit}

We prove the main theorem, restated for convenience.

\begin{theorem} Fix any $n$ and a degree distribution $\cD$ such that $\EX[d]$ and $\EX[d^{4/3}]$ are bounded.
Then
\[
\lim_{n \rightarrow \infty} \frac{1}{n}\EX\left[\sum_{i=1}^n {X_{i,n} \choose 2}\right]   =   \frac{1}{2 (\EX[d])^2}   \sum_{t_1=1}^\infty \sum_{t_2=t_1}^\infty \sum_{t_3=t_1}^\infty t_1(t_1-1)t_2t_3 f(t_1)f(t_2)f(t_3) \in (0,\infty).
\]
\end{theorem}

\begin{proof} By linearity of expectation and the fact that all degrees are chosen identically,
$\frac{1}{n}\EX\left[\sum_{i=1}^n {X_{i,n} \choose 2}\right] = \EX[X_{1,n}(X_{1,n}-1)/2]$.
So we only need the limit of the expression in \Lem{wedge-limit}. We first show the second summation
is negligible.
\begin{eqnarray*}
& & \sum_{t_1 > \log n} \sum_{t_2 = \delta t_1}^{M(n)} \sum_{t_3 = \delta t_1}^{M(n)} d^2_1 t_2 t_3 f(t_1) f(t_2) f(t_3) \\
& \leq & \delta^{-2/3} \sum_{t_1 > \log n} \sum_{t_2 = \delta t_1}^{M(n)} \sum_{t_3 = \delta t_1}^{M(n)} d^{4/3}_1 d^{4/3}_2 d^{4/3}_3 f(t_1) f(t_2) f(t_3) \\
& \leq & \delta^{-2/3} \sum_{d > \log n} d^{4/3} f(d)
\end{eqnarray*}
Since $\EX[d^{4/3}] = \sum_{t = 1}^\infty t^{4/3} f(t)$ is finite, $\lim_{n \rightarrow \infty} \sum_{t > \log n} t^{4/3} f(t) = 0$.
For the first triple summation in \Lem{wedge-limit}, again, we can upper bound the term by $O(\EX[d^{4/3}]^3)$.
It is also nonnegative and monotonically increasing with $n$, so by the monotone convergence theorem,
it converges to limit given in the theorem statement.

\end{proof}

\vspace*{-.2in}
\section{Experimental Analysis}
\label{sec:experiments}

We experimentally show the theoretical analysis of~\Sec{careful} does a reasonable job of capturing the expected performance
of \algo{} on ECM graphs.  Though real-world graphs likely have additional structure, this partially validates the good practical performance of \algo{} in practice.

We generated ECM graphs of various sized based on a power-law degree distribution
with power exponent $\alpha=2.4$ (which guarantees a finite $\frac{4}{3}$ moment for the degree distribution).
Figure~\ref{fig:trunc-study}(a) shows the average value of
$\sum_{i=1}^n {X_{i,n} \choose 2}$, total work over all buckets,
computed over 10 Monte Carlo trials (i.e., taken as an approximation
of $\E[\sum_{i=1}^n {X_{i,n} \choose 2}] $) for ECM graphs of various
sizes up to $n = 80$ million. Degrees are truncated at $\sqrt{n}$.
The ECM use power law reference degree distributions for $\alpha = 2.3$, where \algo{} runs in superlinear time, and for
$\alpha = 2.4$, where it runs in linear time.
Figure~\ref{fig:trunc-study}(a) also shows the theoretical linear bound
on the overall expected work $\E[\sum_{i=1}^n {X_{i,n} \choose 2}] $ for a power-law degree distribution with $\alpha=2.4$.
The constant is at most:
\[
\lim_{n\to \infty} \frac{1}{n}\E\left[\sum_{i=1}^n {X_{i,n} \choose 2}\right] \equiv
C =  \frac{1}{2 (\E[D])^2}   \sum_{d_1=0}^\infty \sum_{d_2=d_1}^\infty \sum_{d_3=d_1}^\infty d_1(d_1-1)d_2d_3 f(d_1)f(d_2)f(d_3) \approx 0.687935.
\]
As $n\rightarrow \infty$, we would anticipate from \Thm{limit}
that the value of $\E[\sum_{i=1}^n \sum_{i=1}^n {X_{i,n} \choose 2}] $ approximated by Monte Carlo trials should approach  $n C$, also shown in
Figure~\ref{fig:trunc-study}(a).  Figure~\ref{fig:trunc-study}(b) shows the ratio of work to number of nodes $n$.
For power law distributions with $\alpha = 2$, this ratio is not a constant.  But by $\alpha = 2.4$, the factor is leveling off below $1$.

\noindent
\begin{figure}[h]
\begin{center}
\begin{tabular}{cc}
\includegraphics*[width=3.0in]{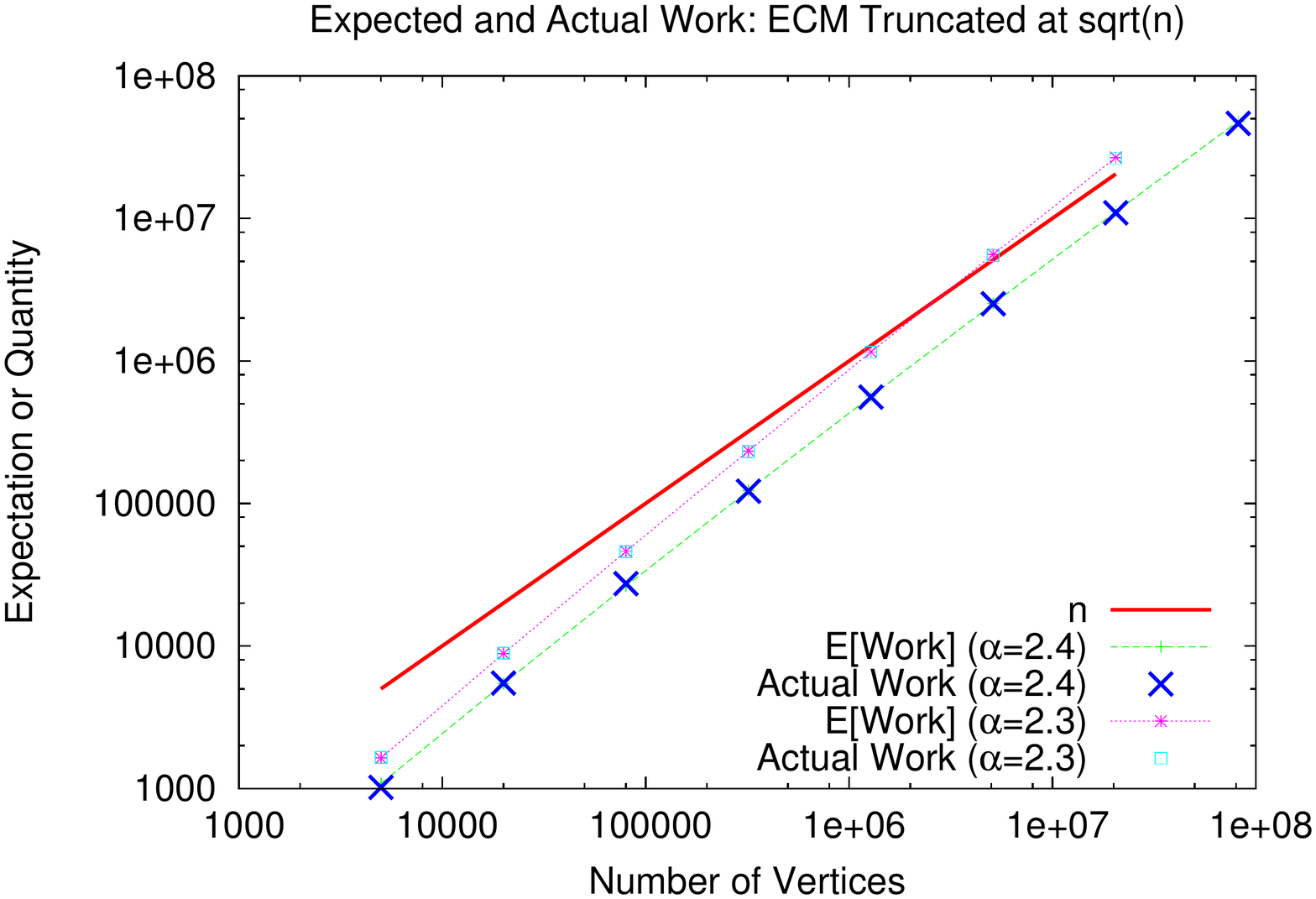} &
\includegraphics*[width=3.0in]{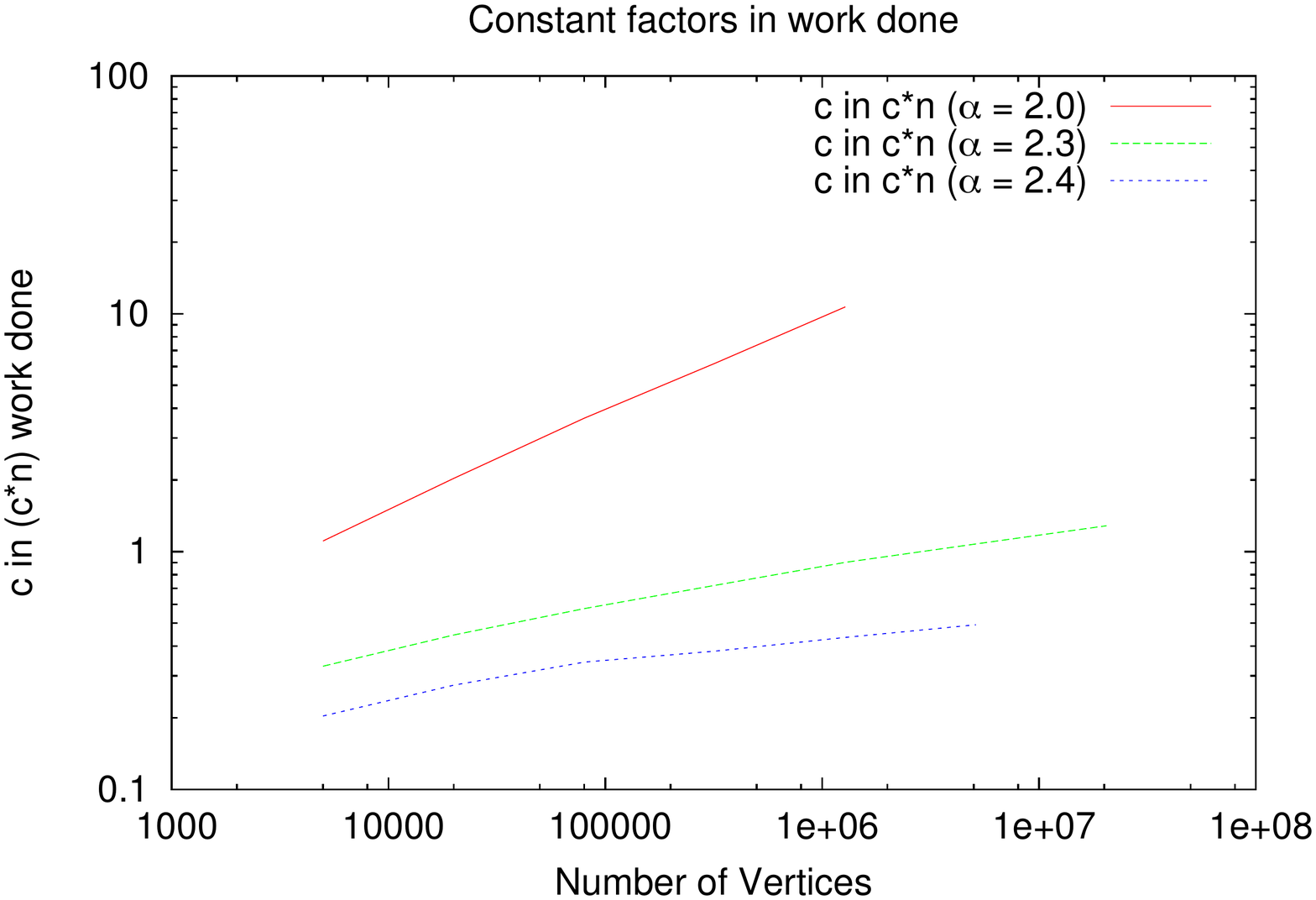}\vspace*{-.2in}\\
\vspace{-.25in}
(a) & (b)
\end{tabular}
\end{center}
\caption{Experimental results \label{fig:trunc-study} with $n$-node ECM graphs with degrees $\le n^{1/2}$
drawn from a power-law distribution with exponent $\alpha=2.3$ or $2.4$.  The red solid line shows $n$.
The green dashed line shows a theoretical  bound $0.687935 n$ on the expected number
of pairs in buckets for an ECM graph with $n$ vertices and exponent $2.4$.  The blue crosses shows the average value of pairs observed in 10 generations of ECM graphs.
The magenta line and blue boxes show the same for exponent $2.3$.
(b) Experimental values of the ratio of work to $n$ for power law exponents $2$, $2.3$, and $2.4$.}
\end{figure}

\newpage
\section*{Acknowledgments}
This work was funded under the Sandia National Laboratories Laboratory Directed Research and
Development (LDRD) program. We thank Ali Pinar for suggestions on improving the presentation.

\bibliographystyle{alpha}
\bibliography{triangles}

\appendix

\section{Proof of \Lem{mart}}

\begin{proof} Consider the sequence $X'_1, X'_2, \ldots, X'_k$
of i.i.d. Bernoulli random variables with $\EX[X'_i] = \alpha$. We will
shortly prove that for any $t>0$, $\Pr[\sum_{i=1}^k X_i < t] \leq \Pr[\sum_{i=1}^k X'_i < t]$.
Given this, we just apply a multiplicative Chernoff bound (Theorem 4.2 of \cite{MoRa95})
for $\sum_{i=1}^k X'_i$ with $\mu = \alpha k$.
Hence, $\Pr[\sum_{i=1}^k X_i < \alpha k \delta] < \exp(-\alpha (1-\delta)^2/2)$.

For convenience, we show the contrapositive $\Pr[\sum_{i=1}^k X_i \geq t] \geq \Pr[\sum_{i=1}^k X'_i \geq t]$.
This is proven by induction on $k$. First, the base case.
Since $X_1$ and $X'_1$ are Bernoulli random variables, it suffices to show that
$\Pr[X_1 = 1] \geq \Pr[X'_1 = 1] = \alpha$, which holds by assumption.

Now for the induction step. Assume for all $t>0$ and some index $j$, $\Pr[\sum_{i=1}^j X_i \geq t] \geq \Pr[\sum_{i=1}^j X'_i \geq t]$. We prove this for $j+1$.
Let $\cE$ denote the event $\sum_{i=1}^j X_i \geq t$,
and $\cE'$ be the (disjoint) event $\sum_{i=1}^j X_i \in [t-1,t)$.
Let $\mathbb{I}(A)$ denote the indicator function of event $A$.
Because $X_i$
is a $0$-$1$ random variable, we get
\begin{eqnarray*}
\Pr\left[\sum_{i=1}^{j+1} X_i \geq t\right] = \Pr[\cE] + \Pr[\cE' \wedge (X_{j+1} = 1)]
&=&  \Pr[\cE] + \EX[ \mathbb{I}(\cE')\mathbb{I}(X_{j+1} = 1)]\\
& &  \Pr[\cE] + \EX\{ \EX[\mathbb{I}(\cE')\mathbb{I}(X_{j+1} = 1)| Y_1,\ldots,Y_j ] \}
\end{eqnarray*}

Observe that $\sum_{i=1}^j X_i$ only depends on $Y_1, \ldots, Y_j$ so that $\mathbb{I}(\cE')$ is a constant
in the conditional expectation
\begin{eqnarray*}
\EX[\mathbb{I}(\cE')\mathbb{I}(X_{j+1} = 1)| Y_1,\ldots,Y_j ] &=& \mathbb{I}(\cE')\EX[\mathbb{I}(X_{j+1} = 1)| Y_1,\ldots,Y_j ]\\&=& \mathbb{I}(\cE') \Pr[X_{j+1} = 1| Y_1,\ldots,Y_j ] \\&\geq& \mathbb{I}(\cE') \alpha,
\end{eqnarray*}
where $\Pr[X_{j+1} = 1| Y_1,\ldots,Y_j ] \geq \alpha$ by the lemma assumption.

Let us denote (for any $s>0$) $\Pr[\sum_{i=1}^j X_i \geq s]$ by $p_s$
and $\Pr[\sum_{i=1}^j X'_i \geq s]$ by $p'_s$. The above gives
\begin{eqnarray*}
\Pr\left[\sum_{i=1}^{j+1} X_i \geq t\right] & \geq & p_t +  \alpha\EX[\mathbb{I}(\cE')] \\
& = & p_t + (p_{t-1} - p_t)\alpha = p_{t-1} \alpha + p_t (1-\alpha)\\
& \geq & p'_{t-1} \alpha + p'_t (1-\alpha) \ \ \ \textrm{(using induction hypothesis and $\alpha \in [0,1]$)}\\
& = & p'_t + (p'_{t-1} - p'_t)\alpha \\
&=&  \Pr\left[\sum_{i=1}^{j} X'_i \geq t\right] + \Pr\left[\Big(\sum_{i=1}^{j} X'_i \in [t-1,t)\Big) \wedge (X'_{j+1}=1)\right]\\&=&    \Pr\left[\sum_{i=1}^{j+1} X'_i \geq t\right]
\end{eqnarray*} 
\end{proof}

\section{Proofs of Tightness} \label{app:tight}

We need a technical claim give a lower bound for probabilities of edges
falling in a bucket.

\begin{claim} \label{clm:Y} Let $d_v > 3$. Consider vertices $v,w,w'$ ($w \neq w'$) and let
$c$ be a sufficiently large constant. If $\min(d_w,d_{w'}) > c d_v$,
then $\EX[Y_{v,w}Y_{v,w'}] = \Omega(d^2_vd_wd_{w'}/m^2)$.
\end{claim}

\begin{proof} The random variable $Y_{v,w}Y_{v,w'}$ is $1$ if $(v,w)$, $(v,w')$
are edges and the degrees of $w$ and $w'$ are less than that of $v$.
As before, we will start
the matching process by matching stubs of $v$. We partition the stubs into two groups
denoted by $B_w$ and $B_{w'}$, and start by matching stubs in $B_w$.
We set $|B_w| = \lfloor d_v/3 \rfloor$. What is the probability that a stub in $B_w$
connects with a $w$-stub? This is at least $1 - (1-d_w/2m)^{\lfloor d_v/3 \rfloor}
= \Omega(d_v d_w/m)$.

Condition on any matching of the stubs in $B_w$. What is the probability that a stub in $B_{w'}$
matches with a $w'$-stub? Since $\min(|B_{w'}|,d_{w'}) \geq 2|B_w|$, this probability
is at least $1 - (1-d_{w'}/4m)^{\lfloor d_v/3 \rfloor} = \Omega(d_v d_{w'}/m)$.

Now condition on any matching of the $v$-stubs. The number of unmatched stubs connected
to $w$ is at least $d_{w}/2$ (similarly for $w'$). The remaining stubs connect
according to a standard configuration model. For the remaining degree sequence,
the total number of stubs is $2\tilde{m} = 2m - 2d_v$. For sufficiently large $m$,
$d_n \leq \sqrt{m}/4 \leq \sqrt{\tilde{m}}/2$. Hence, we can use \Lem{lower-tail}
(and a union bound) to argue that the probability that the final degrees
of $w$ and $w'$ are at least $d_v$ is $\Omega(1)$. Multiplying all the bounds
together, the probability $Y_{v,w}Y_{v,w'} = 1$ is $\Omega(d^2_vd_wd_{w'}/m^2)$.
\end{proof}

We prove \Clm{tight}.

\begin{proof} Note that when $\alpha > 2$, then $m = O(n)$.
We start with the arguments in the proof of \Lem{xv}. Applying \Clm{Y} for vertex $v$ such that $d_v > 3$,
\begin{eqnarray*} \EX[X_v(X_v - 1)] = \sum_w\sum_{w' \neq w} \EX[Y_{v,w} Y_{v,w'}] & \geq &
\sum_{\substack{w: \\ d_w \geq c d_v}} \sum_{\substack{w \neq w':\\ d_{w'}
\geq c d_v}} \EX[Y_{v,w} Y_{v,w'}] \\
& \gg & m^{-2} d^2_v
\sum_{\substack{w: \\ d_w \geq c d_v}} \sum_{\substack{w \neq w':\\ d_{w'}
\geq c d_v}} d_w d_{w'} \\
& \geq & m^{-2} d^2_v (\sum_{\substack{w: \\ d_w \geq c d_v}} d_w)^2 - m^{-2} d^2_v \sum_w d^2_w
\end{eqnarray*}
The latter part, summed over all $v$ is at most
$$m^{-2}(\sum_v d^2_v)^2 \leq m^{-2}(\max_v d_v \sum_v d_v)^2 \lessdot m$$
Now we focus on the former part. Choose $v$ so that $cd_v \leq d_n/2$, and let $2^r$ be the largest
power of $2$ greater than $cd_v$. (Note that $r \leq \log_2 d_n - 1$.)
We bound $\sum_{w: d_w \geq c d_v} d_w \geq \sum_{w: d_w \geq 2^r} d_w
\gg \sum_{k = r}^{\log_2 d_n-1} 2^k n/2^{k(\alpha-1)}$.
This is $\sum_{k = r}^{\log_2 d_n-1} n/2^{k(\alpha-2)}$, which is convergent
when $\alpha > 2$. Hence, it is at least $\Omega(n2^{-r(\alpha-2)}) = \Omega(nd^{-(\alpha-2)}_v)$.

We sum over all (appropriate $v$).
\begin{eqnarray*} \sum_{v: 3 < d_v \leq d_n/2c} m^{-2} d^2_v (\sum_{\substack{w: \\ d_w \geq c d_v}} d_w)^2
& \gg & (n/m)^2  \sum_{v: 3 < d_v \leq d_n/2c} d^2_v d^{-(2\alpha-4)}_v\\
& = & (n/m)^2  \sum_{v: 3 < d_v \leq d_n/2c} d^{6-2\alpha}_v
\gg (n/m)^2 \sum_{k = 2}^{\lfloor \log_2n - \log_2(2c)\rfloor} n 2^{k(7-3\alpha)}
\end{eqnarray*}
When $\alpha < 7/3$, the sum is divergent. Noting that $m = \Theta(n)$, we bound by $\Omega(nd^{7-3\alpha}_n)$. Overall, we lower bound the running time \algo{} by
$\sum_{v: 3 < 	d_v \leq d_n/2c} \EX[X_v(X_v - 1)] $, which is $\Omega(nd^{7-3\alpha}_n - m)$.
For $\alpha < 7/3$, this is $\Omega(nd^{7-3\alpha}_n)$, matching the upper bound in \Cor{power}.

\end{proof}

\section{The running time of \algo{} for Chung-Lu graphs} \label{app:cl}

\begin{theorem} \label{thm:Acohen} Consider a Chung-Lu graph distribution
with $n$ vertices over a degree distribution $f_1, f_2, \ldots, f_n$.
The expected running time of \algo{} is given by
$O(m + n(\sum_v {d_v}^{4/3})^3)$.
\end{theorem}

We remind the reader that the Chung-Lu (CL) model involves
inserting edge $(i,j)$ with probability $d_i d_j/2m$ for all unordered
pairs $(i,j)$.
We need to prove \Clm{arms} for the Chung-Lu model. \Thm{Acohen} will
then follow directly using the arguments in \Sec{cohen}.

We first state Bernstein's inequality.

\begin{theorem} \label{thm:Abernstein} [Bernstein's inequality] Let $X_1, X_2, \ldots, X_k$
be zero-mean independent random variables. Suppose $|X_i| \leq M$ almost surely.
Then for all positive $t$,
$$ \Pr[\sum_{i=1}^k X_i > t] \leq \exp\Big(-\frac{t^2/2}{\sum_i \EX[X^2_i] + Mt/3}\Big) $$
\end{theorem}

We now prove some tail bounds about degrees of vertices. The basic form of these statements
is the probability that degree of vertex $v$ deviates by a constant factor of $d_v$
is $\exp(-\Omega(d_v))$. We state in terms of conditional events for easier
application later. We use $\beta$ to denote a sufficiently small constant.

\begin{claim} \label{clm:Aupper-tail} Let $d \geq 2$. Suppose $v$ is a vertex such that $d_v \leq d$ and $e,e'$ be two
pairs. Let $\cE$ be the event that $e,e'$ are present, and $D_v$ be the random variable
denoting the degree of $v$.  For sufficiently small constant $\beta$,
$$ \Pr[D_v > 3d | \cE] < \exp(-\beta d) $$
\end{claim}

\begin{proof} All edges are inserted independently. So the occurrence of edge $e'' \neq e,e'$ is completely
independent of $\cE$. Let $\delta(v)$ be the set of all pairs involving $v$
and $\hat{\delta}(v) = \delta(v) \setminus \{e,e'\}$.
We express $D_v = \sum_{h \in \delta(v)} C_h$, where $C_{h}$ is the indicator
random variable for edge $h$ being present. Let $\hat{D}_v = \sum_{h \in \hat{\delta{v}}} C_h$.
Note that $\EX[\hat{D}_v] \leq \EX[D_v] = d_v \leq d$. Set $C'_{h} = C_{h} - \EX[C_{h}]$,
so
$$ \Pr[\hat{D}_v - \EX[\hat{D}_v] > d] = \Pr[\sum_{h \in \hat{\delta}(v)} (C_h - \EX[C_h]) > d] =
\Pr[\sum_{h \in \hat{\delta}(v)} C'_h > d] $$
We wish to apply Bernstein's inequality to the $C'_h$ random variables.
Observe that $\EX[C'_{h}] = 0$, and $|C'_{h}| \leq 1$. Setting $\EX[C_{h}] = \mu$, note that
$$ \EX[(C'_{h})^2] = \EX[(C_{h} - \mu)^2] = \EX[C^2_{h}] - \mu\EX[C_{h}] + \mu^2 = \EX[C_{h}].$$
So $\sum_{h \in \hat{\delta}(v)} \EX[(C'_{h})^2] =$  $\sum_{h \in \hat{\delta}(v)} \EX[C_{h}] = \EX[\hat{D}_v]$
$\leq d$.
By Bernstein's inequality (\Thm{Abernstein}),
\begin{eqnarray*} \Pr[\hat{D}_v - \EX[\hat{D}_v] > d]  = \Pr[\sum_{h \in \hat{\delta}(v)} C'_h > d] & \leq & \exp\Big(-\frac{d^2/2}{\sum_{h \in \hat{\delta}(v)} \EX[(C'_{h})^2] + d/3}\Big)  \\
& \leq & \exp\Big(-\frac{d^2/2}{d + d/3}\Big) = \exp(-3d/8)
\end{eqnarray*}
None of these random variables depend on the event $\cE$, so we get that
$\Pr[\hat{D}_v - \EX[\hat{D}_v] > d \ | \ \cE] \leq \exp(-3d/8)$.
Suppose $\hat{D}_v \leq \EX[\hat{D}_v] + d \leq 2d$.
We always have $D_v \leq \hat{D}_v + 2$ and hence $D_v \leq 3d$
(using the bound that $d \geq 2$). Hence, $\Pr[D_v > 3d | \cE] < \exp(-3d/8)$.
We only require $\beta < 3/8$.
\end{proof}

\begin{claim} \label{clm:Alower-tail} Suppose $v$ is a vertex such that $d_v \geq 4$ and $e,e'$ be two
pairs. Let $\cE$ be the event that $e,e'$ are present, and $D_v$ be the random variable
denoting the degree of $v$. For sufficiently small constant $\beta$,
$$ \Pr[D_v < d_v/3 | \cE] < \exp(-\beta d_v) $$
\end{claim}

\begin{proof} This proof is almost identical to the previous one. Again, we express $D_v = \sum_{h \in \delta(v)} C_h$, where $C_{h}$ is the indicator
random variable for edge $h$ being present. Let $\hat{D}_v = \sum_{h \in \hat{\delta{v}}} C_h$.
We have $\hat{D}_v \geq D_v - 2$, so $\EX[\hat{D}_v] \geq \EX[D_v] - d_v/2$ $= d_v/2$ (using the bound $d_v \geq 4$).
Applying a multiplicative Chernoff bound to $\hat{D}_v$,
$$ \Pr[\hat{D}_v < 2\EX[\hat{D}_v]/3] < \exp(-d_v/36) $$
Since $\hat{D}_v$ is completely independent of $\cE$, we can condition
on $\cE$ to get the same bound. Suppose $D_v < d_v/3$. Since $D_v \geq \hat{D}_v$
and $d_v \leq 2\EX[\hat{D}_v]$, we get $\hat{D}_v < 2\EX[\hat{D}_v]/3$.
So the even $D_v < d_v/3 | \cE$ is contained in $\hat{D}_v < 2\EX[\hat{D}_v]/3 | \cE$,
completing the proof. We require $\beta < 1/36$.
\end{proof}

Finally, we need a simple claim about the second moment of sums of independent
random variables.
\begin{claim}\label{clm:Aindicator-square}
Let $X = \sum_i X_i$ be a sum of independent positive random variables with
$X_i = O(1)$ for all $i$ and $\EX[X] = O(1)$. Then $\EX[X^2] = O(1)$.
\end{claim}

\begin{proof}
By linearity of expectation,
\begin{multline*}
\EX\left[X^2\right]  = \EX\Big[\big(\sum_i X_i\big)^2\Big] =
\sum_i \EX\left[X_i^2\right] + 2\sum_{i < j} \EX[X_i] \EX[X_j] \\
\leq  \sum_i O\left(\EX[{X_i}]\right) + \Big(\sum_i \EX[X_i]\Big)^2 = O(1).
\end{multline*}
\end{proof}

We prove the analogue of \Clm{arms}.

\begin{claim} \label{clm:Aarms} Consider vertices $v,w,w'$ ($w \neq w'$).
\begin{itemize}
\item If $d_v \leq 4$, then $\EX[X^2_v] = O(1)$.
\item $\EX[Y_{v,w}Y_{v,w'}] \leq d^2_vd_wd_{w'}/4m^2$.
\item If $d_w \leq d_v/10$ and $d_v \geq 4$, then $\EX[Y_{v,w}Y_{v,w'}] \leq 2\exp(-\beta d_v) d^2_vd_wd_{w'}/4m^2$.
\end{itemize}
\end{claim}

\begin{proof}  Defining $\hat{X}_v = \sum_w C_{v,w}$, we have $X_v \leq \hat{X}_v$.
Since these are all positive random variables, $X^2_v \leq \hat{X}^2_v$.
Applying \Clm{Aindicator-square}, $\EX[\hat{X}^2_v] = O(1)$. That completes
the first part.

For the second part, we use the trivial bound of $Y_{v,w}Y_{v,w'} \leq C_{v,w}C_{v,w'}$.
Taking expectations and using independence, $\EX[Y_{v,w}Y_{v,w'}] \leq C_{v,w}C_{v,w'} = d^2_vd_wd_{w'}/4m^2$.

The third case is really the interesting one. The quantity $\EX[Y_{v,w}Y_{v,w'}]$
is the probability that both $Y_{v,w}$ and $Y_{v,w'}$ are $1$. For this to
happen, we definitely required both $(v,w)$ and $(v,w')$ to be present as edges.
Call this event $\cE$. We also require (at the very least) the degree of $v$
to be at most the degree of $w$ (otherwise the edge $(v,w)$ will not
be put in $v$'s bin.) Call this event $\cF$. The event $Y_{v,w}Y_{v,w'} = 1$
is contained in $\cE \cap \cF$. Using conditional probabilities,
$\Pr(\cE \cap \cF) = \Pr(\cF | \cE) \Pr(\cE)$. Note that $\Pr(\cE) = d^2_vd_wd_{w'}/4m^2$.

Let $D_v, D_w$ denote
the degrees of $v$ and $w$. Let $\cF_v$ denote the event $D_v < d_v/3$
and $\cF_w$ denote event $D_w > 3d_v/10$. If neither of these events
happens, then $D_w \leq 3d_v/10 < d_v/3 \leq D_v$. So $\cF$ cannot happen.
Hence, $(\cF | \cE)$ is contained in $(\cF_v \cup \cF_w | \cE)$.
By the union bound, $\Pr(\cF_v \cup \cF_w | \cE) \leq \Pr(\cF_v | \cE) + \Pr(\cF_w | \cE)$.
Applying \Clm{Alower-tail} to the latter and \Clm{Aupper-tail} to the former,
we bound $\Pr(\cF | \cE) \leq 2\exp(-\beta d_v)$.
\end{proof}

As mentioned earlier, we can now execute the arguments in \Sec{cohen} to prove \Thm{Acohen}.

\end{document}